\documentclass[draftcls,onecolumn]{IEEEtran}
\usepackage{amsmath}
\usepackage{amssymb}
\usepackage{amsfonts}
\usepackage{slashbox}
\usepackage{graphicx}
\usepackage{cite, amsmath, amsfonts, amssymb, psfrag, epsfig}
\newtheorem{definition}{{Definition}}
\newtheorem{theorem}{{Theorem}}
\newtheorem{lemma}{{Lemma}}

\newtheorem{corollary}{{Corollary}}

\DeclareMathAlphabet{\mathpzc}{OT1}{pzc}{m}{it}

\begin{document}

\title{A Random Variable Substitution Lemma With Applications to Multiple Description Coding}


\author{Jia Wang, Jun Chen, Lei Zhao, Paul Cuff, and Haim Permuter
\thanks{Jia Wang is with the Department
of Electronic Engineering, Shanghai Jiao Tong University, Shanghai
200240, China  (email: jiawang@sjtu.edu.cn).}
\thanks{Jun Chen is with the Department
of Electrical and Computer Engineering, McMaster University,
Hamilton, ON L8S 4K1, Canada  (email: junchen@ece.mcmaster.ca).}
\thanks{Lei Zhao and Paul Cuff are with the Department
of Electrical Engineering, Stanford University,
Stanford, CA 94305, USA  (email: \{leiz,cuff\}@stanford.edu).}
\thanks{Haim Permuter is with the Department
of Electrical and Computer Engineering, Ben-Gurion University of the Negev,
Beer-Sheva 84105, Israel  (email: haimp@bgu.ac.il).}}

\maketitle

\begin{abstract}
We establish a random variable substitution lemma and use it to
investigate the role of refinement layer in multiple description
coding, which clarifies the relationship among several existing
achievable multiple description rate-distortion regions.
Specifically, it is shown that the El Gamal-Cover (EGC) region is
equivalent to the EGC* region (an antecedent version of the EGC
region) while the Venkataramani-Kramer-Goyal (VKG) region (when
specialized to the 2-description case) is equivalent to the
Zhang-Berger (ZB) region. Moreover, we prove that for multiple
description coding with individual and hierarchical distortion
constraints, the number of layers in the VKG scheme can be
significantly reduced when only certain weighted sum rates are
concerned. The role of refinement layer in scalable coding (a
special case of multiple description coding) is also studied.
\end{abstract}

\begin{keywords}
Contra-polymatroid, multiple description coding, rate-distortion region, scalable coding, successive refinement,
\end{keywords}

\section{Introduction}

A fundamental problem of multiple description coding is to
characterize the rate-distortion region, which is the set of all
achievable rate-distortion tuples. El Gamal and Cover (EGC) obtained
an inner bound of the 2-description rate-distortion region, which is
shown to be tight for  the no excess rate case by Ahlswede
\cite{Ahlswede85}. Zhang and Berger (ZB) \cite{Zhang-Berger87}
derived a different inner bound of the 2-description rate-distortion
region and showed that it contains rate-distortion tuples not
included in the EGC region. The EGC region has an antecedent
version, which is sometimes referred to as the EGC* region. The EGC*
region was shown to be tight for the quadratic Gaussian case by
Ozarow \cite{Ozarow80}. However, the EGC* region has been largely
abandoned in view of the fact that it is contained in the EGC region
\cite{Zhang-Berger87}. Other work on the 2-description problem can
be found in \cite{Feng05,Fu02,Lastras06,Zamir99}. Recent years have
seen growth of interest in the general $L$-description problem
\cite{Pradhan04,Puri05,TC08,CMD09,Venkataramani03}. In particular,
Venkataramani, Kramer, and Goyal (VKG) \cite{Venkataramani03}
derived an inner bound of the $L$-description rate-distortion
region. It is well understood that for the 2-description case both
the EGC region and the ZB region subsume the EGC* region while all
these three regions are contained in the VKG region; moreover, the
reason that one region contains another is simply because more
layers are used. Indeed, the ZB scheme has one more common
description layer than the EGC* scheme while the EGC scheme and the
VKG scheme have one more refinement layer than the EGC* scheme and
the ZB scheme, respectively. Although it is known
\cite{Zhang-Berger87} that the EGC* scheme can be strictly improved
via the inclusion of a common description layer, it is still unclear
whether the refinement layer has the same effect. We shall show that
in fact the EGC region is equivalent to the EGC* region and the VKG
region is equivalent to the ZB region; as a consequence, the
refinement layer can be safely removed.

An important special case of the 2-description problem is called
scalable coding, also known as successive refinement\footnote{The
notion of successive refinement is sometimes used in the more
restrictive no rate loss scenario.}. The rate-distortion region of
scalable coding has been characterized by Koshelev \cite{Koshelev80}
\cite{Koshelev81}, Equitz and Cover \cite{Equitz91} for the no rate
loss case and by Rimoldi \cite{Rimoldi94} for the general case. In
scalable coding, the second description is not required to
reconstruct the source; instead, it serves as a refinement layer to
improve the first description. However, it is clearly of interest to
know whether the refinement layer itself in an optimal scalable
coding scheme can be useful, i.e., whether one can achieve a
nontrivial distortion using the refinement layer alone. This problem
is closely related, but not identical, to multiple description
coding with no excess rate.

To the end of understanding the role of refinement layer in multiple
description coding as well as scalable coding, we need the following
random variable substitution lemma.
\begin{lemma}\label{le:le1}
Let $U$, $V$, and $W$ be jointly distributed random variables taking
values in finite sets $\mathcal {U}$, $\mathcal{V}$, and
$\mathcal{W}$, respectively. There exist a  random
variable $Z$, taking values in a finite set $\mathcal{Z}$ with
$|\mathcal{Z}|\le|\mathcal{V}||\mathcal{W}|-1$, and a function $f:
\mathcal{V} \times \mathcal{Z} \to \mathcal{W}$ such that
\begin{enumerate}
\item $Z$ is independent of $V$;

\item $W=f(V,Z)$;

\item $U-(V,W)-Z$ form a Markov chain.
\end{enumerate}
\end{lemma}
The proof of Lemma \ref{le:le1} is given in Appendix \ref{sec:app1}.
Roughly speaking, this lemma states that one can remove random
variable $W$ by introducing random variable $Z$ and deterministic
function $f$. It will be seen in the context of multiple description
coding that $Z$ can be incorporated into other random variables due
to its special property, which results in a reduction of the number
of random variables.

The remainder of this paper is devoted to the applications of the
random variable substitution lemma to multiple description coding
and scalable coding. In Section \ref{sec:Mainresult}, we show that
the EGC region is equivalent to the EGC* region and the ZB region
includes the EGC region. We examine the general $L$-description
problem in Section \ref{sec:Mannychannels}. It is shown that the
final refinement layer in the VKG scheme can be removed. This result
implies that the VKG region, when specialized to the 2-description
case, is equivalent to the ZB region. Furthermore, we prove that for
multiple description coding with individual and hierarchical
distortion constraints, the number of layers in the VKG scheme can
be significantly reduced when only certain weighted sum rates are
concerned. We study scalable coding with an emphasis on the role of
refinement layer in Section \ref{sec:successiverefinement}. Section
\ref{sec:conclusion} contains some concluding remarks.

\section{Applications to the 2-description case}
\label{sec:Mainresult}

We shall first give a formal definition of the multiple description rate-distortion region. Let $\{X(t)\}_{t=1}^\infty$ be an i.i.d. process with marginal distribution $p_X$ on $\mathcal{X}$, and $d:\mathcal{X}\times\hat{\mathcal{X}}\rightarrow[0,\infty)$ be a distortion measure, where $\mathcal{X}$ and $\hat{\mathcal{X}}$ are finite sets. Define $\mathcal{I}_L=\{1,\cdots,L\}$ for any positive integer $L$.

\begin{definition}
A rate-distortion tuple $(R_1,\cdots,R_L,D_{\mathcal{K}},\emptyset\subset\mathcal{K}\subseteq\mathcal{L})$ is said to be achievable if for any $\epsilon>0$, there exist
encoding functions $f^{(n)}_k:\mathcal{X}^n\rightarrow\mathcal{C}^{(n)}_k$, $k\in\mathcal{I}_L$, and decoding functions $g^{(n)}_{\mathcal{K}}:\prod_{k\in\mathcal{K}}\mathcal{C}^{(n)}_k\rightarrow\hat{\mathcal{X}}^n$, $\emptyset\subset\mathcal{K}\subseteq\mathcal{I}_L$,
such that
\begin{align*}
&\frac{1}{n}\log|\mathcal{C}^{(n)}_k|\leq R_k+\epsilon,\quad k\in\mathcal{I}_L,\\
&\frac{1}{n}\sum\limits_{t=1}^n\mathbb{E}[d(X(t),\hat{X}_{\mathcal{K}}(t))]\leq D_{\mathcal{K}}+\epsilon,\quad \emptyset\subset\mathcal{K}\subseteq\mathcal{I}_L,
\end{align*}
for all sufficiently large $n$, where $\hat{X}^n_{\mathcal{K}}=g^{(n)}_{\mathcal{K}}(f^{(n)}_k(X^n),k\in\mathcal{K})$. The multiple description rate-distortion region $\mathcal{RD}_{\text{MD}}$ is the set of all achievable rate-distortion tuples.
\end{definition}

We shall focus on the 2-description case (i.e., $L=2$) in this section. The following two inner bounds of $\mathcal{RD}_{\text{MD}}$ are attributed to El Gamal and Cover.

The EGC* region $\mathcal{RD}_{\text{EGC*}}$ is the convex closure
of the set of quintuples $(R_1,R_2,D_{\{1\}},D_{\{2\}},D_{\{1,2\}})$ for which there exist
auxiliary random variables $X_{\{1\}}$ and $X_{\{2\}}$, jointly distributed with $X$, and
functions $\phi_{\mathcal{K}},\emptyset\subset\mathcal{K}\subseteq\{1,2\}$, such that
\begin{align*}
&R_k\geq I(X;X_{\{k\}}),\quad k\in\{1,2\},\\
&R_1+R_2\geq
I(X;X_{\{1\}},X_{\{2\}})+I(X_{\{1\}};X_{\{2\}}),\\
&D_{\{k\}}\geq
\mathbb{E}[d(X,\phi_{\{i\}}(X_{\{i\}}))],\quad k\in\{1,2\},\\
&D_{\{1,2\}}\geq \mathbb{E}[d(X,\phi_{\{1,2\}}(X_{\{1\}},X_{\{2\}}))].
\end{align*}

The EGC region $\mathcal{RD}_{\text{EGC}}$ is the convex closure
of the set of quintuples $(R_1,R_2,D_{\{1\}},D_{\{2\}},D_{\{1,2\}})$ for which there exist auxiliary
random variables $X_{\mathcal{K}}$, $\emptyset\subset\mathcal{K}\subseteq\{1,2\}$, jointly distributed with
$X$, such that
\begin{align}
&R_k\geq I(X;X_{\{k\}}),\quad k\in\{1,2\},\label{eq:EGC1}\\
&R_1+R_2\geq
I(X;X_{\{1\}},X_{\{2\}},X_{\{1,2\}})+I(X_{\{1\}};X_{\{2\}}),\label{eq:EGC3}\\
&D_\mathcal{K}\geq \mathbb{E}[d(X,X_{\mathcal{K}})],\quad \emptyset\subset\mathcal{K}\subseteq\{1,2\}.\label{eq:EGC4}
\end{align}

To see the connection between these two inner bounds, we shall write the EGC region in an alternative form. It can be verified that the EGC region is equivalent
to the set of
quintuples $(R_1,R_2,D_{\{1\}},D_{\{2\}},D_{\{1,2\}})$ for which there exist auxiliary
random variables $X_{\mathcal{K}}$, $\emptyset\subset\mathcal{K}\subseteq\{1,2\}$, jointly distributed with
$X$, and functions $\phi_{\mathcal{K}},\emptyset\subset\mathcal{K}\subseteq\{1,2\}$, such that
\begin{align*}
&R_k\geq I(X;X_{\{k\}}),\quad k\in\{1,2\},\\
&R_1+R_2\geq
I(X;X_{\{1\}},X_{\{2\}},X_{\{1,2\}})+I(X_{\{1\}};X_{\{2\}}),\\
&D_{\{k\}}\geq
\mathbb{E}[d(X,\phi_{\{k\}}(X_{\{k\}}))],\quad k\in\{1,2\},\\
&D_{\{1,2\}}\geq \mathbb{E}[d(X,\phi_{\{1,2\}}(X_{\{1\}},X_{\{2\}},X_{\{1,2\}}))].
\end{align*}
It is easy to see from this alternative form of the EGC region that the only difference from the EGC* region is the additional random variable $X_{\{1,2\}}$, which
corresponds to a refinement layer; by setting $X_{\{1,2\}}$ to be constant (i.e, removing the refinement layer), we recover the EGC* region. Therefore, the EGC* region is contained in the EGC region. It is natural to ask whether the refinement layer leads to a strict improvement. The answer turns out to be negative as shown by the following theorem, which states that the two regions are in fact equivalent.

\begin{theorem}\label{th:th1}
$\mathcal{RD}_{\text{EGC*}}=\mathcal{RD}_{\text{EGC}}$.
\end{theorem}
\begin{proof}
In view of the fact that $\mathcal{RD}_{\text{EGC*}}\subseteq\mathcal{RD}_{\text{EGC}}$, it suffices to prove $\mathcal{RD}_{\text{EGC}}\subseteq
\mathcal{RD}_{\text{EGC*}}$.

For any fixed $p_{XX_{\{1\}}X_{\{2\}}X_{\{1,2\}}}$, the region specified by (\ref{eq:EGC1})-(\ref{eq:EGC4}) has two vertices
\begin{align*}
&v_1: (R_1(v_1),R_2(v_1),D_{\{1\}}(v_1),D_{\{2\}}(v_1),D_{\{1,2\}}(v_1)),\\
&v_2: (R_1(v_2),R_2(v_2),D_{\{1\}}(v_2),D_{\{2\}}(v_2),D_{\{1,2\}}(v_2)),
\end{align*}
where
\begin{align*}
&R_1(v_1)=I(X;X_{\{1\}}),\\
&R_2(v_1)=I(X;X_{\{2\}},X_{\{1,2\}}|X_{\{1\}})+I(X_{\{1\}};X_{\{2\}}),\\
&R_1(v_2)=I(X;X_{\{1\}},X_{\{1,2\}}|X_{\{2\}})+I(X_{\{1\}};X_{\{2\}}),\\
&R_2(v_2)=I(X;X_{\{2\}}),\\
&D_{\mathcal{K}}(v_1)=D_{\mathcal{K}}(v_2)=\mathbb{E}[d(X,X_{\mathcal{K}})],\quad\emptyset\subset\mathcal{K}\subseteq\{1,2\}.
\end{align*}
We just need to show that both vertices are contained in the EGC* region. By symmetry, we shall only consider vertex $v_1$.

It follows from Lemma \ref{le:le1} that there exist a random variable
$Z$, jointly distributed with $(X,X_{\{1\}},X_{\{2\}},X_{\{1,2\}})$, and a function $f$
such that
\begin{enumerate}
\item $Z$ is independent of $(X_{\{1\}},X_{\{2\}})$;

\item $X_{\{1,2\}}=f(X_{\{1\}},X_{\{2\}},Z)$;

\item $X-(X_{\{1\}},X_{\{2\}},X_{\{1,2\}})-Z$ form a Markov chain.
\end{enumerate}
By the fact that $X-(X_{\{1\}},X_{\{2\}},X_{\{1,2\}})-Z$ form a Markov chain and that $X_{\{1,2\}}$ is a deterministic function of $(X_{\{1\}},X_{\{2\}},Z)$, we have
\begin{align*}
I(X;X_{\{2\}},X_{\{1,2\}}|X_{\{1\}})&=I(X;X_{\{2\}},X_{\{1,2\}},Z|X_{\{1\}})\\
&=I(X;X_{\{2\}},Z|X_{\{1\}}).
\end{align*}
Moreover, since $Z$ is independent of $(X_{\{1\}},X_{\{2\}})$, it follows that
\begin{align*}
I(X_{\{1\}};X_{\{2\}})=I(X_{\{1\}};X_{\{2\}},Z).
\end{align*}
By setting $X'_{\{2\}}=(X_{\{2\}},Z)$, we can rewrite the coordinates of $v_1$ as
\begin{align*}
&R_1(v_1)=I(X;X_{\{1\}}),\\
&R_2(v_1)=I(X;X_{\{1\}},X'_{\{2\}})+I(X_{\{1\}};X'_{\{2\}}),\\
&D_{\{1\}}(v_1)=\mathbb{E}[d(X,\phi_{\{1\}}(X_{\{1\}}))],\\
&D_{\{2\}}(v_1)=\mathbb{E}[d(X,\phi_{\{2\}}(X'_{\{2\}}))],\\
&D_{\{1,2\}}(v_1)=\mathbb{E}[d(X,\phi_{\{1,2\}}(X_{\{1\}},X'_{\{2\}}))],
\end{align*}
where $\phi_{\{1\}}(X_{\{1\}})=X_{\{1\}}$, $\phi_{\{2\}}(X'_{\{2\}})=X_{\{2\}}$, and $\phi_{\{1,2\}}(X_{\{1\}},X'_{\{2\}})=f(X_{\{1\}},X_{\{2\}},Z)=X_{\{1,2\}}$.
Therefore, it is clear that vertex $v_1$ is contained in the EGC* region. The proof is complete.
\end{proof}

{\em Remark}: It is worth noting that the proof of Theorem \ref{th:th1} implicitly provides cardinality bounds for the auxiliary random variables of the EGC* region.

Now we shall proceed to discuss the ZB region, which is also an inner bound of $\mathcal{RD}_{\text{MD}}$. The ZB region $\mathcal{RD}_{\text{ZB}}$ is the set of quintuples  $(R_1,R_2,D_{\{1\}},D_{\{2\}},D_{\{1,2\}})$ for which there exist auxiliary
random variables $X_{\emptyset}$, $X_{\{1\}}$, and $X_{\{2\}}$, jointly distributed with
$X$,  and functions $\phi_{\mathcal{K}},\emptyset\subset\mathcal{K}\subseteq\{1,2\}$, such that
\begin{align*}
&R_k\geq I(X;X_{\emptyset},X_{\{k\}}),\quad k\in\{1,2\},\\
&R_1+R_2\geq 2I(X;X_{\emptyset})+I(X;X_{\{1\}},X_{\{2\}}|X_{\emptyset})+I(X_{\{1\}};X_{\{2\}}|X_{\emptyset}),\\
&D_{\{k\}}\geq \mathbb{E}[d(X,\phi_{\{k\}}(X_{\emptyset},X_{\{k\}}))],\quad k\in\{1,2\},\\
&D_{\{1,2\}}\geq \mathbb{E}[d(X,\phi_{\{1,2\}}(X_{\emptyset},X_{\{1\}},X_{\{2\}}))].
\end{align*}
Note that the ZB region is a convex set. It is easy to see from the
definition of the ZB region that its only difference from the EGC*
region is the additional random variable $X_{\emptyset}$, which
corresponds to a common description layer; by setting
$X_{\emptyset}$ to be constant (i.e., removing the common
description layer), we recover the EGC* region. Therefore, the EGC*
region is contained in the ZB region, and the following result is an
immediate consequence of Theorem \ref{th:th1}.

\begin{corollary}\label{cor:cor1}
$\mathcal{RD}_{\text{EGC}}\subseteq\mathcal{RD}_{\text{ZB}}$.
\end{corollary}

\emph{Remark}: Since the ZB region contains rate-distortion tuples
not in the EGC region as shown in \cite{Zhang-Berger87}, the
inclusion can be strict.

\section{Applications to the $L$-description case}
\label{sec:Mannychannels}

The general $L$-description problem turns out to be considerably more complex than the 2-description case. The difficulty might be attributed to the following fact. For any two non-empty subsets of  $\{1,2\}$, either one contains the other or they are disjoint; however, this is not true for subsets of $\mathcal{I}_L$ when $L>2$. Indeed, this tree structure of distortion constraints is a fundamental feature that distinguishes the 2-description problem from the general $L$-description problem.

The VKG region  \cite{Venkataramani03}, which is a natural combination and extension of the EGC region and the ZB region, is an inner bound of the $L$-description rate-distortion region. We shall show that the final refinement layer in the VKG scheme is dispensable, which implies that the VKG region, when specialized to the 2-description case, coincides with the ZB region. We formulate the problem of multiple description coding with individual and hierarchical distortion constraints, which is a special case of tree-structured distortion constraints, and show that in this setting the number of layers
in the VKG scheme can be significantly reduced when only certain weighted sum rates are concerned. It is worth noting that the VKG scheme is not the only scheme known for the $L$-description problem. Indeed, there are several other schemes in the literature \cite{Pradhan04,Puri05,TC08} which can outperform the VKG scheme in certain scenarios where the distortion constraints do no exhibit a tree structure. However, the VKG scheme remains to be the most natural one for tree-structured distortion constraints.

We shall adopt the notation in \cite{Venkataramani03}. For any set $\mathcal{A}$, let $2^{\mathcal{A}}$
be the power set of $\mathcal{A}$. Given a collection of sets
$\mathcal{B}$, we define $X_{(\mathcal{B})}=\{X_{\mathcal{A}}:\mathcal{A} \in
\mathcal{B}\}$. Note that $X_{\emptyset}$ (which is a random variable) should not be confused with $X_{(\emptyset)}$ (which is interpreted as a constant).
We use $R_{\mathcal{K}}$ to denote  $\sum\nolimits_{k \in \mathcal{K}}{R_k}$ for $\emptyset\subset\mathcal{K}\subseteq\mathcal{I}_L$.

The VKG region $\mathcal{RD}_{\text{VKG}}$ is the set of rate-distortion tuples  $(R_1,\cdots,R_L,D_{\mathcal{K}},\emptyset\subset\mathcal{K}\subseteq\mathcal{I}_L)$ for which there exist auxiliary
random variables $X_{\mathcal{K}}$, $\mathcal{K}\subseteq\mathcal{I}_L$, jointly distributed with
$X$,  and functions $\phi_{\mathcal{K}},\emptyset\subset\mathcal{K}\subseteq\mathcal{I}_L$, such that
\begin{align}
&R_{\mathcal{K}}  \geq \psi(\mathcal{K}),\quad \emptyset\subset\mathcal{K}\subseteq\mathcal{I}_L,\label{eq:rate}\\
&D_{\mathcal{K}} \geq \mathbb{E}[ {d_{\mathcal{K}}
(X,\phi_{\mathcal{K}}(X_{(2^{\mathcal{K}})}) )}],\quad \emptyset\subset\mathcal{K}\subseteq\mathcal{I}_L,\label{eq:distortion}
\end{align}
where
\begin{align*}
\psi(\mathcal{K})=( {| \mathcal{K} | - 1}
) I(X;X_{\emptyset} ) - H(X_{(2^{\mathcal{K}})} |X) +
\sum\limits_{\mathcal{A} \subseteq \mathcal{K}} {H(X_{\mathcal{A}}
|X_{(2^{\mathcal{A}} - \{ \mathcal{A}\} )} )}.
\end{align*}
Note that the VKG region is a convex set. In fact, reference \cite{Venkataramani03} contains a weak version and a strong version of the VKG region, and the one given here is in a slightly different form from those in \cite{Venkataramani03}. Specifically, one can get the weak version in \cite{Venkataramani03} by replacing (\ref{eq:distortion}) with $D_{\mathcal{K}} \geq \mathbb{E}[d_{\mathcal{K}}
(X, X_{\mathcal{K}})]$, and get the strong version in \cite{Venkataramani03} by replacing (\ref{eq:distortion}) with $D_{\mathcal{K}} \geq \mathbb{E}[d_{\mathcal{K}}
(X, \phi_{\mathcal{K}}(X_{\mathcal{K}}))]$. It is easy to verify that the strong version is equivalent to the one given here while both of them are at least as large as the weak version; moreover, all these three versions are equivalent when $L=2$.

We shall first give a structural characterization of the VKG region.
\begin{lemma}\label{le:contra}
For any fixed $p_{XX_{(2^{\mathcal{I}_L})}}$, the rate region $\{(R_1,\cdots,R_L):R_{\mathcal{K}}\geq\psi(\mathcal{K}),\emptyset\subset\mathcal{K}\subseteq\mathcal{I}_L\}$
is a contra-polymatroid.
\end{lemma}
\begin{proof}
See Appendix \ref{app:contra}.
\end{proof}

Note that the random variable $X_{\mathcal{I}_L}$ corresponds to the
final refinement layer in the VKG scheme. Now we proceed to show
that this refinement layer can be removed. Define the VKG* region
$\mathcal{RD}_{\text{VKG*}}$ as the VKG region with
$X_{\mathcal{I}_L}$ set to be a constant.
\begin{theorem}\label{th:VKG*}
$\mathcal{RD}_{\text{VKG*}}=\mathcal{RD}_{\text{VKG}}$.
\end{theorem}
\begin{proof}
The proof is given in Appendix
\ref{sec:VKG*}.
\end{proof}

A direct consequence of Theorem \ref{th:VKG*} is that the VKG region, when specialized to the 2-description case, is equivalent to the ZB region.
\begin{corollary}\label{cor:cor2}
For the 2-description
problem, $\mathcal{RD}_{\text{ZB}}=\mathcal{RD}_{\text{VKG}}$.
\end{corollary}

{\em Remark}: For the 2-description VKG region, the cardinality bound for $X_{\emptyset}$ can be derived by invoking the supporting lemma \cite{Csiszar81} while all the other auxiliary random variables can be assumed, with no loss of generality, to be defined on the reconstruction alphabet $\hat{\mathcal{X}}$. Therefore, one can deduce cardinality bounds for the auxiliary random variables of the ZB region by leveraging Corollary \ref{cor:cor2}.

We can see that for the VKG* region, the number of auxiliary random variables is exactly the same as the number of distortion constraints. Intuitively, the number of auxiliary random variables can be further reduced if we remove certain distortion constraints. Somewhat surprisingly, we shall show that in some cases the number of auxiliary random variables can be significantly less than the number of distortion constraints.

\begin{figure}[t]
   \begin{center}
      \epsfig{file=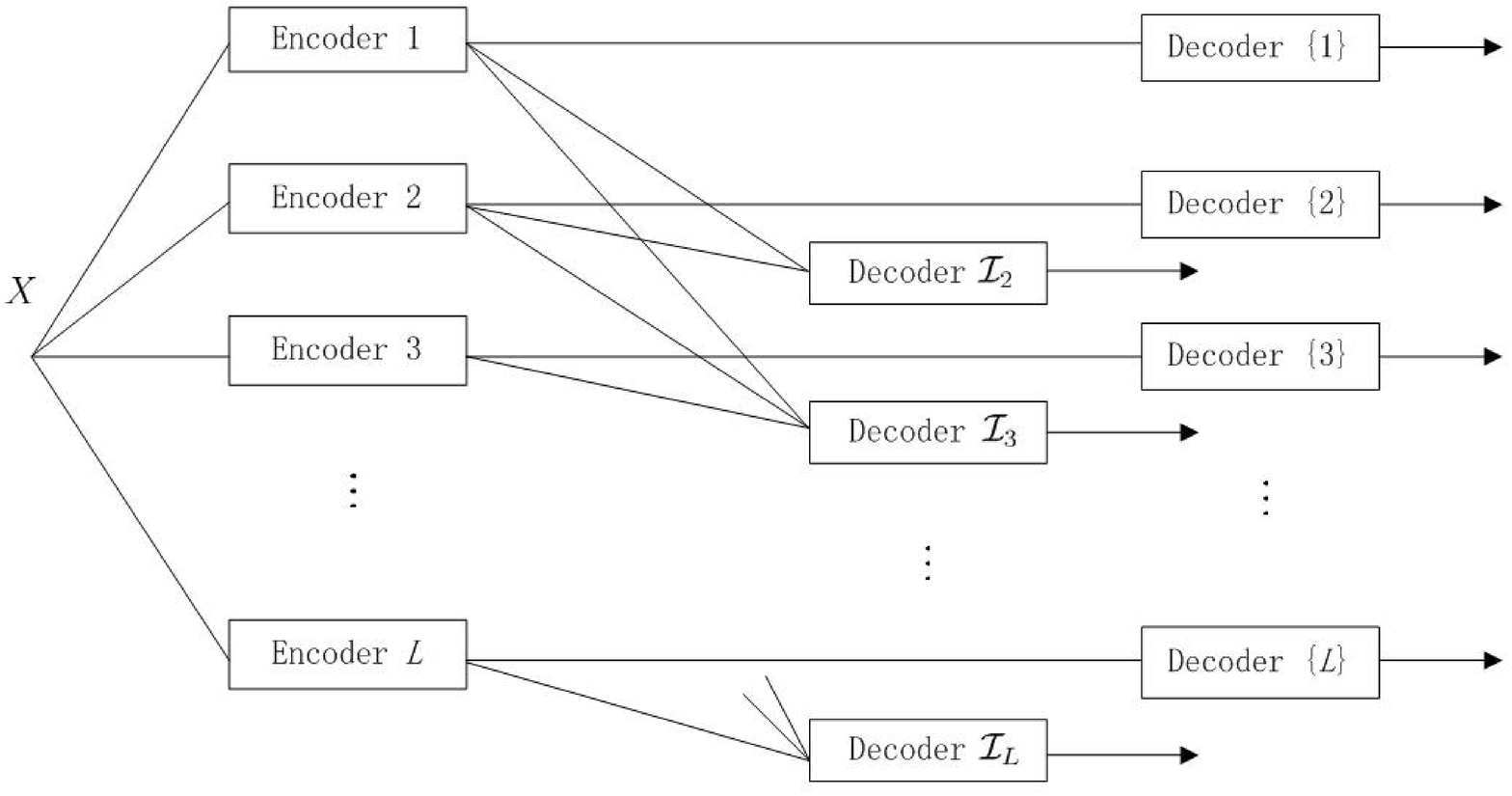, width = 4.5in}
      \caption{Multiple description coding with individual and hierachical distortion constraints.}
      \label{fig:fig_IH}
    \end{center}
\end{figure}

For any nonnegative integer $k$, define $\mathcal{H}_k=\emptyset$ if $k=0$, $\mathcal{H}_k=\{\{1\}\}$ if $k=1$, and $\mathcal{H}_k=\{\{1\},\cdots,\{k\},\mathcal{I}_2,\cdots,\mathcal{I}_k\}$ if $k\geq 2$. Multiple description coding with individual and hierachical distortion
constraints (see Fig. \ref{fig:fig_IH}) refers to the scenario where only the following distortion constraints: $D_{\mathcal{K}}$, $\mathcal{K}\in\mathcal{H}_L$, are imposed.
Specializing the VKG region to this setting, we can define the VKG region for multiple description coding with individual and hierachical distortion
constraints $\mathcal{RD}_{\text{IH-VKG}}$ as the set of rate-distortion tuples  $(R_1,\cdots,R_L,D_{\mathcal{K}},\mathcal{K}\in\mathcal{H}_L)$ for which there exist
auxiliary random variables $X_{\mathcal{K}}$, $\mathcal{K}\subseteq\mathcal{I}_L$, jointly distributed with
$X$,  and functions $\phi_{\mathcal{K}}$, $\mathcal{K}\in\mathcal{H}_L$, such that
\begin{align*}
&R_{\mathcal{K}}  \geq \psi(\mathcal{K}),\quad \emptyset\subset\mathcal{K}\subseteq\mathcal{I}_L,\\
&D_{\mathcal{K}} \geq \mathbb{E}[ {d_{\mathcal{K}}
(X,\phi_{\mathcal{K}}(X_{(2^{\mathcal{K}})}) )}],\quad \mathcal{K}\in\mathcal{H}_L.
\end{align*}
Define $\mathcal{R}_{\text{IH-VKG}}(D_{\mathcal{K}},\mathcal{K}\in\mathcal{H}_L)=\{(R_1,\cdots,R_L):(R_1,\cdots,R_L,D_{\mathcal{K}},\mathcal{K}\in\mathcal{H}_L)\in\mathcal{RD}_{\text{IH-VKG}}\}$.
It is observed in \cite{Chen09} that for the quadratic Gaussian case, the number of auxiliary random variables can be significantly reduced when only certain supporting hyperplanes of $\mathcal{R}_{\text{IH-VKG}}(D_{\mathcal{K}},\mathcal{K}\in\mathcal{H}_L)$ are concerned. We shall show that this phenomenon is not restricted to the quadratic Gaussian case.

\begin{theorem}\label{th:MDindividual}
For any $\alpha_1\geq\cdots\alpha_L\geq 0$, we have
\begin{align}
&\min\limits_{(R_1,\cdots,R_L)\in\mathcal{R}_{\text{IH-VKG}}(D_{\mathcal{K}},\mathcal{K}\in\mathcal{H}_L)}\sum\limits_{k=1}^L\alpha_kR_k\nonumber\\
&=\min\limits_{p_{X_{\emptyset}X_{\{1\}}\cdots X_{\{L\}}|X},\phi_{\mathcal{K}},\mathcal{K}\in\mathcal{H}_L}\sum\limits_{k=1}^L\alpha_k[I(X;X_{\emptyset})+I(X,\{X_{\{i\}}\}_{i=1}^{k-1};X_{\{k\}}|X_{\emptyset})],\label{eq:min}
\end{align}
where the minimization in (\ref{eq:min}) is over $p_{X_{\emptyset}X_{\{1\}}\cdots X_{\{L\}}|X}$, and $\phi_{\mathcal{K}}$, $\mathcal{K}\in\mathcal{H}_L$, subject to the constraints
\begin{align*}
&D_{\{k\}}\geq\mathbb{E}[d(X,\phi_{\{k\}}(X_{\emptyset},X_{\{k\}}))],\quad k\in\mathcal{I}_L,\\
&D_{\mathcal{I}_{k}}\geq\mathbb{E}[d(X,\phi_{\mathcal{I}_k}(X_{\emptyset},X_{\{1\}},\cdots,X_{\{k\}}))],\quad k\in\mathcal{I}_L-\{1\}.
\end{align*}
\end{theorem}
\begin{proof}
The proof of Theorem \ref{th:MDindividual} is given in Appendix
\ref{sec:MDindividual}.
\end{proof}

\begin{corollary}\label{cor:cor3}
For any $\alpha_1\geq\cdots\alpha_L\geq 0$, we have
\begin{align}
&\min\limits_{(R_1,\cdots,R_L)\in\mathcal{R}_{\text{IH-VKG}}(D_{\mathcal{K}},\mathcal{K}\in\mathcal{H}_L)}\sum\limits_{k=1}^L\alpha_kR_k\nonumber\\
&=\min\limits_{p_{X_{\emptyset}X_{(\mathcal{H}_L)}|X}}\sum\limits_{k=1}^L\alpha_k[I(X;X_{\emptyset})+I(X_{(\mathcal{H}_{k-1})};X_{\{k\}}|X_{\emptyset})+I(X;X_{\{k\}},X_{\mathcal{I}_k}|X_{\emptyset},X_{(\mathcal{H}_{k-1})})],\label{eq:min2}
\end{align}
where the minimization in (\ref{eq:min2}) is over $p_{X_{\emptyset}X_{(\mathcal{H}_L)}|X}$ subject to the constraints
\begin{align*}
&D_{\mathcal{K}}\geq\mathbb{E}[d(X,X_{\mathcal{K}})],\quad \mathcal{K}\in\mathcal{H}_L.
\end{align*}
\end{corollary}
\begin{proof}
See Appendix \ref{app:cor3}.
\end{proof}

{\em Remark}: It should be noted that $X_{\mathcal{K}}$, $\mathcal{K}\in\mathcal{H}_L$, in (\ref{eq:min2}) are defined on the reconstruction alphabet $\hat{\mathcal{X}}$; moreover, for $X_{\emptyset}$ in (\ref{eq:min2}), the cardinality bound can be easily derived by invoking the support lemma \cite{Csiszar81}. In view of the proof of Corollary \ref{cor:cor3}, one can derive cardinality bounds for the auxiliary random variables in (\ref{eq:min}) by leveraging the cardinality bounds for the auxiliary random variables in (\ref{eq:min2}). This explains why ``$\min$" instead of ``$\inf$" is used in (\ref{eq:min}).

A special case of multiple description coding with individual and hierachical distortion
constraints  is called multiple description coding with individual and central distortion
constraints \cite{Chen09,Wang07}, where only the individual distortion constraints $D_{\{k\}}$, $k\in\mathcal{I}_L$, and the central distortion constraint $D_{\mathcal{I}_L}$ are imposed.
Let $\mathcal{G}_L=\{\{1\},\cdots,\{L\},\mathcal{I}_L\}$. We can define the VKG region for multiple description coding with individual and central distortion
constraints $\mathcal{RD}_{\text{IC-VKG}}$ as the set of rate-distortion tuples  $(R_1,\cdots,R_L,D_{\mathcal{K}},\mathcal{K}\in\mathcal{G}_L)$ for which there exist
auxiliary random variables $X_{\mathcal{K}}$, $\mathcal{K}\subseteq\mathcal{I}_L$, jointly distributed with
$X$,  and functions $\phi_{\mathcal{K}}$, $\mathcal{K}\in\mathcal{G}_L$, such that
\begin{align*}
&R_{\mathcal{K}}  \geq \psi(\mathcal{K}),\quad \emptyset\subset\mathcal{K}\subseteq\mathcal{I}_L,\\
&D_{\mathcal{K}} \geq \mathbb{E}[ {d_{\mathcal{K}}
(X,\phi_{\mathcal{K}}(X_{(2^{\mathcal{K}})}) )}],\quad \mathcal{K}\in\mathcal{G}_L.
\end{align*}
Define $\mathcal{R}_{\text{IC-VKG}}(D_{\mathcal{K}},\mathcal{K}\in\mathcal{G}_L)=\{(R_1,\cdots,R_L):(R_1,\cdots,R_L,D_{\mathcal{K}},\mathcal{K}\in\mathcal{G}_L)\in\mathcal{RD}_{\text{IC-VKG}}\}$.
The following result is a simple consequence of Theorem \ref{th:MDindividual} and Corollary \ref{cor:cor3}.

\begin{corollary}
$\mathcal{RD}_{\text{IC-VKG}}$ is equivalent to
the set of rate-distortion tuples  $(R_1,\cdots,R_L,D_{\mathcal{K}},\mathcal{K}\in\mathcal{G}_L)$ for which there exist auxiliary
random variables $X_{\emptyset}$, $X_{\{k\}}$, $k\in\mathcal{I}_L$, jointly distributed with
$X$,  and functions $\phi_{\mathcal{K}}$, $\mathcal{K}\in\mathcal{G}_L$, such that
\begin{align*}
&R_{\mathcal{K}}  \geq | \mathcal{K} | I(X;X_{\emptyset} ) - H(\{X_{\{k\}}\}_{k\in\mathcal{K}}|X,X_{\emptyset}) +\sum\limits_{k\in\mathcal{K}}H(X_{\{k\}}|X_{\emptyset}),\quad \emptyset\subset\mathcal{K}\subseteq\mathcal{I}_L,\\
&D_{\{k\}} \geq \mathbb{E}[ {d(X,\phi_{\{k\}}(X_{\emptyset},X_{\{k\}}))}],\quad k\in\mathcal{I}_L,\\
&D_{\mathcal{I}_L}\geq\mathbb{E}[d(X,\phi_{\mathcal{I}_L}(X_{\emptyset},X_{\{1\}},\cdots,X_{\{L\}}))].
\end{align*}
$\mathcal{RD}_{\text{IC-VKG}}$ is also equivalent to
the set of rate-distortion tuples  $(R_1,\cdots,R_L,D_{\mathcal{K}},\mathcal{K}\in\mathcal{G}_L)$ for which there exist auxiliary
random variables $X_{\emptyset}$, $X_{\mathcal{K}}$, $\mathcal{K}\in\mathcal{G}_L$, jointly distributed with
$X$,  and functions $\phi_{\mathcal{K}}$, $\mathcal{K}\in\mathcal{G}_L$, such that
\begin{align*}
&R_{\mathcal{K}} \geq | \mathcal{K} | I(X;X_{\emptyset} ) - H(\{X_{\{k\}}\}_{k\in\mathcal{K}}|X,X_{\emptyset}) +\sum\limits_{k\in\mathcal{K}}H(X_{\{k\}}|X_{\emptyset}),\quad \emptyset\subset\mathcal{K}\subset\mathcal{I}_L,\\
&R_{\mathcal{I}_L}\geq LI(X;X_{\emptyset} ) - H(\{X_{\{k\}}\}_{k\in\mathcal{I}_L}|X,X_{\emptyset}) +\sum\limits_{k=1}^LH(X_{\{k\}}|X_{\emptyset})+I(X;X_{\mathcal{I}_L}|X_{\emptyset},\{X_{\{k\}}\}_{k\in\mathcal{I}_L}),\\
&D_{\mathcal{K}} \geq \mathbb{E}[ {d(X,X_{\mathcal{K}})}],\quad \mathcal{K}\in\mathcal{G}_L.
\end{align*}
Moreover, for any $(\alpha_1,\cdots,\alpha_L)\in\mathbb{R}^L_+$, let $\pi$ be a permutation on $\mathcal{I}_L$ such that $\alpha_{\pi(1)}\geq\cdots\geq\alpha_{\pi(L)}$; we have
\begin{align}
&\min\limits_{(R_1,\cdots,R_L)\in\mathcal{R}_{\text{IC-VKG}}(D_{\mathcal{K}},\mathcal{K}\in\mathcal{G}_L)}\sum\limits_{k=1}^L\alpha_kR_k\nonumber\\
&=\min\limits_{p_{X_{\emptyset}X_{\{1\}}\cdots X_{\{L\}}|X},\phi_{\mathcal{K}},\mathcal{K}\in\mathcal{G}_L}\sum\limits_{k=1}^L\alpha_{\pi(k)}[I(X;X_{\emptyset})+I(X,\{X_{\pi(i)}\}_{i=1}^{k-1};X_{\{\pi(k)\}}|X_{\emptyset})]\label{eq:icmin}\\
&=\min\limits_{p_{X_{\emptyset}X_{(\mathcal{G}_L)}|X}}\sum\limits_{k=1}^L\alpha_{\pi(k)}[I(X;X_{\emptyset})+I(X,\{X_{\pi(i)}\}_{i=1}^{k-1};X_{\{\pi(k)\}}|X_{\emptyset})]+\alpha_{\pi(L)}I(X;X_{\mathcal{I}_L}|X_{\emptyset},\{X_{\{k\}}\}_{k\in\mathcal{I}_L}),\label{eq:icmin2}
\end{align}
where the minimization in (\ref{eq:icmin}) is over $p_{X_{\emptyset}X_{\{1\}}\cdots X_{\{L\}}|X}$, and $\phi_{\mathcal{K}}$, $\mathcal{K}\in\mathcal{G}_L$, subject to the constraints
\begin{align*}
&D_{\{k\}}\geq\mathbb{E}[d(X,\phi_{\{k\}}(X_{\emptyset},X_{\{k\}}))],\quad k\in\mathcal{I}_L,\\
&D_{\mathcal{I}_L}\geq\mathbb{E}[d(X,\phi_{\mathcal{I}_L}(X_{\emptyset},X_{\{1\}},\cdots,X_{\{L\}}))],
\end{align*}
while the minimization in (\ref{eq:icmin2}) is over $p_{X_{\emptyset}X_{(\mathcal{G}_L)}|X}$ subject to the constraints
\begin{align*}
&D_{\mathcal{K}}\geq\mathbb{E}[d(X,X_{\mathcal{K}})],\quad \mathcal{K}\in\mathcal{G}_L.
\end{align*}
\end{corollary}

\section{Applications to Scalable Coding}
\label{sec:successiverefinement}

Scalable coding is a special case of the 2-description problem in which the distortion constraint on the second description, i.e., $D_{\{2\}}$, is not imposed.
The scalable coding rate-distortion region $\mathcal{RD}_{\text{SC}}$ is defined as
\begin{align*}
\mathcal{RD}_{\text{SC}}=\{(R_1,R_2,D_{\{1\}},D_{\{1,2\}}):(R_1,R_2,D_{\{1\}},\infty,D_{\{1,2\}})\in\mathcal{RD}_{\text{MD}}\}.
\end{align*}
It is proved in \cite{Rimoldi94} that the quadruple
$(R_1,R_2,D_{\{1\}},D_{\{1,2\}})\in\mathcal{RD}_{\text{SC}}$ if and only if there exist auxiliary random variables $X_{\{1\}}$ and $X_{\{1,2\}}$ jointly distributed with $X$ such that
\begin{align*}
&R_1\geq I(X;X_{\{1\}}),\\
&R_1+R_2\geq I(X;X_{\{1\}},X_{\{1,2\}}),\\
&D_{\{1\}}\geq \mathbb{E}[d(X,X_{\{1\}})],\\
&D_{\{1,2\}}\geq \mathbb{E}[d(X,X_{\{1,2\}})].
\end{align*}
It is clear that one can obtain $\mathcal{RD}_{\text{SC}}$ from $\mathcal{RD}_{\text{EGC}}$ by setting $X_{\{2\}}$ to be a constant.

Since the EGC region is equivalent to the EGC* region, it is not surprising that $\mathcal{RD}_{\text{SC}}$ can be written in an alternative form which resembles the EGC* region. By Lemma \ref{le:le1}, there exist a random variable
$X_{\{2\}}$, jointly distributed with $(X,X_{\{1\}},X_{\{1,2\}})$, and a function $f$,
such that
\begin{enumerate}
\item $X_{\{2\}}$ is independent of $X_{\{1\}}$;

\item $X_{\{1,2\}}=f(X_{\{1\}},X_{\{2\}})$;

\item $X-(X_{\{1\}},X_{\{1,2\}})-X_{\{2\}}$ form a Markov chain.
\end{enumerate}
Therefore, $\mathcal{RD}_{\text{SC}}$ can be written as the set of quadruples
$(R_1,R_2,D_{\{1\}},D_{\{1,2\}})$ for which there exist independent random variables $X_{\{1\}}$ and $X_{\{2\}}$, jointly distributed with $X$, and a function $f$, such that
\begin{align*}
&R_1\geq I(X;X_{\{1\}}),\\
&R_1+R_2\geq I(X;X_{\{1\}},X_{\{2\}}),\\
&D_{\{1\}}\geq \mathbb{E}[d(X,X_{\{1\}})],\\
&D_{\{1,2\}}\geq \mathbb{E}[d(X,f(X_{\{1\}},X_{\{2\}}))].
\end{align*}
It is somewhat interesting to note that a direct verification of the fact that this alternative form of $\mathcal{RD}_{\text{SC}}$ is equivalent to the EGC* region without constraint $D_{\{2\}}$ is not completely straightforward.

Since $D_{\{2\}}$ is not imposed in scalable coding, the second description essentially plays the role of a refinement layer. It is natural to ask whether the refinement layer itself can be useful, i.e., whether one can use the refinement layer alone to achieve a non-trivial reconstruction distortion. However, without further constraint, this problem is essentially the same as the multiple description problem. Therefore, we shall focus on the following special case. Define the minimum scalably
achievable total rate $R(R_1,D_{\{1\}},D_{\{1,2\}})$ with respect to $(R_1,D_{\{1\}},D_{\{1,2\}})$ as
\begin{align*}
R(R_1,D_{\{1\}},D_{\{1,2\}})=\min\{R_1+R_2:(R_1,R_2,D_{\{1\}},D_{\{1,2\}})\in\mathcal{RD}_{\text{SC}}\}.
\end{align*}
It is clear that \cite{Rimoldi94}
\begin{align*}
R(R_1,D_{\{1\}},D_{\{1,2\}})=\mathop{\min}\limits_{I(X;X_{\{1\}})\leq
R_1\atop{\mathbb{E}[d(X,X_{\{1\}})] \leq D_{\{1\}}\atop \mathbb{E}[d(X,X_{\{1,2\}})] \leq D_{\{1,2\}}}} I(X;X_{\{1\}},X_{\{1,2\}}).
\end{align*}
Let $\mathcal{Q}$ denote the convex closure of the set of quintuples $(R_1,R_2,D_{\{1\}},D_{\{2\}},D_{\{1,2\}})$ for which there exist auxiliary
random variables $X_{\mathcal{K}}$, $\emptyset\subset\mathcal{K}\subseteq\{1,2\}$, jointly distributed with
$X$, such that
\begin{align*}
&I(X_{\{1\}};X_{\{2\}})=0,\\
&R_k\geq I(X;X_{\{k\}}),\quad k\in\{1,2\},\\
&R_1+R_2\geq
I(X;X_{\{1\}},X_{\{2\}},X_{\{1,2\}}),\\
&D_\mathcal{K}\geq \mathbb{E}[d(X,X_{\mathcal{K}})],\quad \emptyset\subset\mathcal{K}\subseteq\{1,2\}.
\end{align*}
Note that $\mathcal{Q}$ is essentially the EGC region with an addition constraint $I(X_{\{1\}};X_{\{2\}})=0$ (i.e., $X_{\{1\}}$ and $X_{\{2\}}$ are independent).

\begin{lemma}\label{le:le2}
The EGC region is tight if $R_1+R_2=R(R_1,D_{\{1\}},D_{\{1,2\}})$; more precisely,
\begin{align*}
&\{(R_1,R_2,D_{\{1\}},D_{\{2\}},D_{\{1,2\}})\in\mathcal{RD}_{\text{MD}}:R_1+R_2=R(R_1,D_{\{1\}},D_{\{1,2\}})\}\\
&=\{(R_1,R_2,D_{\{1\}},D_{\{2\}},D_{\{1,2\}})\in\mathcal{Q}:R_1+R_2=R(R_1,D_{\{1\}},D_{\{1,2\}})\}.
\end{align*}
\end{lemma}
\begin{proof}
It is worth noting that this problem is not identical to multiple description coding without excess rate. Nevertheless, Ahlswede's proof technique \cite{Ahlswede85} (also cf. \cite{Tuncel03}) can be directly applied here with no essential change. The details are omitted.
\end{proof}

Let $R(D)$ denote the rate-distortion function, i.e.,
\begin{align*}
R(D)=\min\limits_{p_{\hat{X}|X}:\mathbb{E}[d(X,\hat X)]\leq D}I(X;\hat X).
\end{align*}
Now we proceed to study the minimum achievable $D_{\{2\}}$ in the scenario where $R_1=R(D_{\{1\}})$ and $R_1+R_2=R(R(D_{\{1\}}),D_{\{1\}},D_{\{1,2\}})$. Define
\begin{align*}
D^*_{\{2\}}(D_{\{1\}},D_{\{1,2\}})=\min\limits_{R_1=R(D_{\{1\}})\atop{R_1+R_2=R(R_1,D_{\{1\}},D_{\{1,2\}})\atop(R_1,R_2,D_{\{1\}},D_{\{2\}},D_{\{1,2\}})\in\mathcal{RD}_{\text{MD}}}}
D_{\{2\}}.
\end{align*}
Though $D^*_{\{2\}}(D_{\{1\}},D_{\{1,2\}})$ is in principle computable using Lemma \ref{le:le2}, the calculation is often non-trivial due to the convex hull operation in the definition of the EGC region. We shall show that $D^*_{\{2\}}(D_{\{1\}},D_{\{1,2\}})$ has a more explicit characterization under certain technical conditions.

We need the following definition of weak independence from \cite{BY89}.
\begin{definition}
For jointly distributed random variables $U$ and $V$, $U$ is weakly independent of $V$ if the rows of the stochastic matrix $[p_{U|V}(u|v)]$ are linearly dependent.
\end{definition}

The following lemma can be found in \cite{BY89}.
\begin{lemma}\label{le:weakind}
For jointly distributed random variables $U$ and $V$, there exists a random variable $W$ satisfying
\begin{enumerate}
\item $U-V-W$ form a Markov chain;

\item $U$ and $W$ are independent;

\item $V$ and $W$ are not independent;
\end{enumerate}
if and only if $U$ is weakly independent of $V$.
\end{lemma}
\begin{theorem}\label{th:weakindependence}
If $X$ is not weakly independent of $X_{\{1\}}$ for any $X_{\{1\}}$ induced by $p_{X_{\{1\}}|X}$ that achieves $R(D_{\{1\}})$, then
\begin{align}
D^*_{\{2\}}(D_{\{1\}},D_{\{1,2\}})=\min\limits_{p_{X_{\{1\}}X_{\{2\}}|X},g_1,g_2}\mathbb{E}[d(X,g_1(X_{\{2\}}))],\label{eq:rl}
\end{align}
where the minimization is over $p_{X_{\{1\}}X_{\{2\}}|X}$, $g_1$, and $g_2$  subject to the constraints
\begin{align*}
&I(X_{\{1\}};X_{\{2\}})=0,\\
&I(X;X_{\{1\}})=R(D_{\{1\}}),\\
&I(X;X_{\{1\}},X_{\{2\}})=R(R(D_{\{1\}}),D_{\{1\}},D_{\{1,2\}}),\\
&\mathbb{E}[d(X,X_{\{1\}})]\leq D_{\{1\}},\\
&\mathbb{E}[d(X,g_2(X_{\{1\}},X_{\{2\}}))]\leq D_{\{1,2\}}.
\end{align*}
Here one can assume that $X_{\{2\}}$ is defined on a finite set with cardinality no greater than $|\hat{\mathcal{X}}|^4-|\hat{\mathcal{X}}|$.
\end{theorem}
\begin{proof}
First we shall show that the right-hand side of (\ref{eq:rl}) is achievable. Given any $D_{\{1\}}$ and $D_{\{1,2\}}$ for which there exist auxiliary random variables $X_{\mathcal{K}}$, $\emptyset\subset\mathcal{K}\subseteq\{1,2\}$, jointly distributed with $X$, and a function $g_2$ such that
\begin{align*}
&I(X_{\{1\}};X_{\{2\}})=0,\\
&R(D_{\{1\}})=I(X;X_{\{1\}}),\\
&R(R(D_{\{1\}}),D_{\{1\}},D_{\{1,2\}})=I(X;X_{\{1\}},X_{\{2\}}),\\
&D_{\{1\}}\geq\mathbb{E}[d(X,X_{\{1\}})],\\
&D_{\{1,2\}}\geq\mathbb{E}[d(X,g_2(X_{\{1\}},X_{\{2\}}))],
\end{align*}
we have
\begin{align*}
&R(R(D_{\{1\}}),D_{\{1\}},D_{\{1,2\}})=I(X;X_{\{1\}},X_{\{2\}})+I(X_{\{1\}},X_{\{2\}}),\\
&R(R(D_{\{1\}}),D_{\{1\}},D_{\{1,2\}})-R(D_{\{1\}})=I(X_{\{1\}},X;X_{\{2\}})\geq I(X;X_{\{2\}}).
\end{align*}
Therefore, the quintuple $(R_1,R_2,D_{\{1\}},D_{\{2\}},D_{\{1,2\}})$, where
\begin{align*}
&R_1=R(D_{\{1\}}),\\
&R_2=R(R(D_{\{1\}}),D_{\{1\}},D_{\{1,2\}})-R(D_{\{1\}}),\\
&D_{\{2\}}=\mathbb{E}[d(X,g_1(X_{\{2\}}))],
\end{align*}
is contained in the EGC* region for any function $g_1$. This proves the achievability part.

Now we proceed to prove the converse part. Let $R_1=R(D_{\{1\}})$ and $R_2=R(R(D_{\{1\}}),D_{\{1\}},D_{\{1,2\}})-R(D_{\{1\}})$. Since the VKG region includes the EGC region, Lemma \ref{le:le2} implies that the VKG region is also tight when the total rate is equal to $R(R(D_{\{1\}}),D_{\{1\}},D_{\{1,2\}})$. Therefore, if the quintuple $(R_1,R_2,D_{\{1\}},D_{\{2\}},D_{\{1,2\}})$ is achievable, then there exist auxiliary random variables $X_{\mathcal{K}}$, $\mathcal{K}\subseteq\{1,2\}$, jointly distributed with $X$ such that
\begin{align*}
&R_k\geq I(X;X_{\emptyset},X_{\{k\}}),\quad k\in\{1,2\},\\
&R_1+R_2\geq 2I(X;X_{\emptyset})+I(X;X_{\{1\}},X_{\{2\}},X_{\{1,2\}}|X_{\emptyset})+I(X_{\{1\}};X_{\{2\}}|X_{\emptyset})\\
&D_{\mathcal{K}}\geq \mathbb{E}[d(X,X_{\mathcal{K}})],\quad \emptyset\subset\mathcal{K}\subseteq\{1,2\}.
\end{align*}
By the definition of $R(D_{\{1\}})$ and $R(R(D_{\{1\}}),D_{\{1\}},D_{\{1,2\}})$, we must have
\begin{align*}
&R(D_{\{1\}})=I(X;X_{\emptyset},X_{\{1\}})=I(X;X_{\{1\}}),\\
&R(R(D_{\{1\}}),D_{\{1\}},D_{\{1,2\}})=2I(X;X_{\emptyset})+I(X;X_{\{1\}},X_{\{2\}},X_{\{1,2\}}|X_{\emptyset})+I(X_{\{1\}};X_{\{2\}}|X_{\emptyset})=I(X;X_{\{1\}},X_{\{1,2\}}),
\end{align*}
which implies that
\begin{enumerate}
\item $X$ and $X_{\emptyset}$ are independent;

\item $X-X_{\{1\}}-X_{\emptyset}$ form a Markov chain;

\item $X_{\{1\}}-X_{\emptyset}-X_{\{2\}}$ form a Markov chain;

\item $X-(X_{\{1\}},X_{\{1,2\}})-(X_{\emptyset},X_{\{2\}})$ form a Markov chain;

\item $p_{X_{\{1\}}|X}$ achieves $R(D_{\{1\}})$.
\end{enumerate}
Since $X$ is not weakly independent of $X_{\{1\}}$, it follows from Lemma \ref{le:weakind} that $X_{\emptyset}$ and $X_{\{1\}}$ are independent, which further implies that $X_{\{1\}}$ and $X_{\{2\}}$ are independent.
By Lemma \ref{le:le1}, there exist a random variable $Z$ one $\mathcal{Z}$ with $|\mathcal{Z}|\leq |\hat{\mathcal{X}}|^3-1$ and a function $f$ such that
\begin{enumerate}
\item $Z$ is independent of $(X_{\{1\}},X_{\{2\}})$;

\item $X_{\{1,2\}}=f(X_{\{1\}},X_{\{2\}},Z)$;

\item $X-(X_{\{1\}},X_{\{2\}},X_{\{1,2\}})-Z$ form a Markov chain.
\end{enumerate}
By setting $X'_{\{2\}}=(X_{\{2\}},Z)$, it is easy to verify that
\begin{align*}
&I(X_{\{1\}};X'_{\{2\}})=0,\\
&R(R(D_{\{1\}}),D_{\{1\}},D_{\{1,2\}})=I(X;X_{\{1\}},X'_{\{1,2\}}),\\
&D_{\{2\}}\geq\mathbb{E}[d(X,g_1(X'_{\{2\}})),\\
&D_{\{1,2\}}\geq\mathbb{E}[d(X,g_2(X_{\{1\}},X'_{\{2\}}))],
\end{align*}
where $g_1(X'_{\{2\}})=g_1(X_{\{2\}},Z)=X_{\{2\}}$ and $g_2(X_{\{1\}},X'_{\{2\}})=f(X_{\{1\}},X_{\{2\}},Z)=X_{\{1,2\}}$. The proof is complete.
\end{proof}

Now we give an example for which $D^*_{\{2\}}(D_{\{1\}},D_{\{1,2\}})$ can be calculated explicitly.
\begin{theorem}\label{th:binary}
For a binary symmetric source with Hamming distortion measure,
\begin{align*}
D^*_{\{2\}}(D_{\{1\}},D_{\{1,2\}}) = \frac{1}{2}+D_{\{1,2\}}-D_{\{1\}}
\end{align*}
for $0\leq D_{\{1,2\}}\leq D_{\{1\}}\leq \frac{1}{2}$.
\end{theorem}
\begin{proof}
The proof is given in Appendix \ref{sec:binary}.
\end{proof}

\section{Concluding Remarks}
\label{sec:conclusion}

We have established a random variable substitution lemma and used it
to clarify the relationship among several existing achievable
rate-distortion regions for multiple description coding.

Like many other ideas in information theory, our random variable
substitution lemma finds its seeds in Shannon's pioneering work.
Consider a finite-state channel $p_{Y|XS}$, where the state process
$\{S_t\}_{t=1}^\infty$ is stationary and memoryless. It is well
known that the capacity is given by
\begin{align*}
C=\max\limits_{p_{X|S}}I(X;Y|S)
\end{align*}
when the state process is available at both the transmitter and the receiver. By Lemma \ref{le:le1}, for any $(X,Y,S)$, there exist
a random variable $Z$ on $\mathcal{Z}$ and a function $f:
\mathcal{Z} \times \mathcal{S} \to \mathcal{X}$ such that
\begin{enumerate}
\item $Z$ is independent of $S$;

\item $X=f(S,Z)$;

\item $Y-(X,S)-Z$ form a Markov chain.
\end{enumerate}
Therefore, we have
\begin{align}
C&=\max\limits_{p_{X|S}}I(X;Y|S)\nonumber\\
&=\max\limits_{p_Z,f:\mathcal{Z} \times \mathcal{S} \to \mathcal{X}}I(Z;Y|S).\label{eq:Shannon}
\end{align}
Note that (\ref{eq:Shannon}) is in fact Shannon's capacity formula with channel state information at the transmitter \cite{Shannon58} applied to the case where
the channel state information is also available at the receiver; in this setting, $f(Z,\cdot)$ is sometimes referred to as Shannon's strategy.

\appendices

\section{Proof of Lemma 1}
\label{sec:app1}

Let $Y$ be a random variable independent of $V$ and uniformly distributed over $[0,1]$. It is obvious that for each $v\in\mathcal{V}$ we can find a function $f_v$ satisfying
\begin{align*}
\mathbb{P}(f_v(Y)=w)=p_{W|V}(w|v),\quad w\in\mathcal{W}.
\end{align*}
Now define a function $f$ such that
\begin{align*}
f(v,y)=f_v(y),\quad v\in\mathcal{V}, y\in[0,1].
\end{align*}
It is clear that
\begin{align}
\mathbb{P}(V=v,f(V,Y)=w)=p_{VW}(v,w),\quad v\in\mathcal{V},w\in\mathcal{W}.\label{eq:preservation1}
\end{align}
Note that
\begin{align*}
\mathbb{P}(V=v,f(V,Y)=w)=\mathbb{E}[\mathbb{P}(V=v,f(V,Y)=w|Y)],\quad v\in\mathcal{V},w\in\mathcal{W}.
\end{align*}
It can be shown by invoking the support lemma \cite{Csiszar81} that there exist a finite set $\mathcal{Z}\subset[0,1]$ with $|\mathcal{Z}|\leq|\mathcal{V}||\mathcal{W}|-1$ and a random variable $Z$ on $\mathcal{Z}$, independent of $V$, such that
\begin{align}
\mathbb{P}(V=v,f(V,Z)=w)&=\mathbb{E}[\mathbb{P}(V=v,f(V,Z)=w|Z)]\nonumber\\
&=\mathbb{E}[\mathbb{P}(V=v,f(V,Y)=w|Y)]\nonumber\\
&=\mathbb{P}(V=v,f(V,Y)=w),\quad v\in\mathcal{V},w\in\mathcal{W}.\label{eq:preservation2}
\end{align}
By (\ref{eq:preservation1}) and (\ref{eq:preservation2}), we can see that $p_{VW}$ is preserved if $W$ is set to be equal to $f(V,Z)$. Now we incorporate $U$ into the probability space by setting $p_{U|VWZ}=p_{U|VW}$. It can be readily verified that $p_{UVW}$ is preserved and $U-(V,W)-Z$ indeed form a Markov chain. The proof is complete.

\section{Proof of Lemma \ref{le:contra}}\label{app:contra}

By the definition of contra-polymatroid \cite{Edmonds70}, it
suffices to show that the set function
$\psi:2^{\mathcal{I}_L}\rightarrow\mathbb{R}_+$ satisfies 1)
$\psi(\emptyset)=0$ (normalized), 2)
$\psi(\mathcal{S})\leq\psi(\mathcal{T})$ if
$\mathcal{S}\subset\mathcal{T}$ (nondecreasing), 3)
$\psi(\mathcal{S})+\psi(\mathcal{T})\leq\psi(\mathcal{S}\cup\mathcal{T})+\psi(\mathcal{S}\cap\mathcal{T})$
(supermodular).
\begin{enumerate}
\item Normalized: We have
\begin{align*}
\psi(\emptyset)=-I(X;X_{\emptyset})-H(X_{\emptyset}|X)+H(X_{\emptyset})=0.
\end{align*}

\item Nondecreasing: If $\mathcal{S}\subset\mathcal{T}$, then
\begin{align*}
\psi(\mathcal{T})-\psi(\mathcal{S})&=(|\mathcal{T}|-|\mathcal{S}|)I(X;X_{\emptyset})-H(X_{(2^{\mathcal{T}})}|X)+H(X_{(2^{\mathcal{S}})}|X)+\sum\limits_{\mathcal{A}\in 2^{\mathcal{T}}-2^{\mathcal{S}}}H(X_{\mathcal{A}}|X_{(2^{\mathcal{A}}-\{\mathcal{A}\})})\\
&\geq-H(X_{(2^{\mathcal{T}}-2^{\mathcal{S}})}|X,X_{(2^{\mathcal{S}})})+\sum\limits_{\mathcal{A}\in 2^{\mathcal{T}}-2^{\mathcal{S}}}H(X_{\mathcal{A}}|X_{(2^{\mathcal{A}}-\{\mathcal{A}\})})\\
&\geq-\sum\limits_{k=1}^{|\mathcal{T}|}\sum\limits_{\mathcal{A}\in 2^{\mathcal{T}}-2^{\mathcal{S}},|\mathcal{A}|=k}H(X_{\mathcal{A}}|X,\{X_{\mathcal{B}}\}_{\mathcal{B}\in 2^{\mathcal{T}},|\mathcal{B}|<k})+\sum\limits_{\mathcal{A}\in 2^{\mathcal{T}}-2^{\mathcal{S}}}H(X_{\mathcal{A}}|X_{(2^{\mathcal{A}}-\{\mathcal{A}\})})\\
&\geq-\sum\limits_{k=1}^{|\mathcal{T}|}\sum\limits_{\mathcal{A}\in 2^{\mathcal{T}}-2^{\mathcal{S}},|\mathcal{A}|=k}H(X_{\mathcal{A}}|X_{(2^\mathcal{A}-\{\mathcal{A}\})})+\sum\limits_{\mathcal{A}\in 2^{\mathcal{T}}-2^{\mathcal{S}}}H(X_{\mathcal{A}}|X_{(2^{\mathcal{A}}-\{\mathcal{A}\})})\\
&=0.
\end{align*}

\item Supermodular: We have
\begin{align*}
&(\psi(\mathcal{S}\cup\mathcal{T})-\psi(\mathcal{T}))-(\psi(\mathcal{S})-\psi(\mathcal{S}\cap\mathcal{T}))\\
&=(|\mathcal{S}\cup\mathcal{T}|-|\mathcal{T}|)I(X;X_{\emptyset})-H(X_{(2^{\mathcal{S}\cup\mathcal{T}}-2^{\mathcal{T}})}|X,X_{(2^{\mathcal{T}})})+\sum\limits_{\mathcal{A}\in 2^{\mathcal{S}\cup\mathcal{T}}-2^{\mathcal{T}}}H(X_{\mathcal{A}}|X_{(2^{\mathcal{A}}-\{\mathcal{A}\})})\\
&\quad-(|\mathcal{S}|-|\mathcal{S}\cap\mathcal{T}|)I(X;X_{\emptyset})+H(X_{(2^{\mathcal{S}}-2^{\mathcal{S}\cap\mathcal{T}})}|X,X_{(2^{\mathcal{S}\cap\mathcal{T}})})-\sum\limits_{\mathcal{A}\in 2^{\mathcal{S}}-2^{\mathcal{S}\cap\mathcal{T}}}H(X_{\mathcal{A}}|X_{(2^{\mathcal{A}}-\{\mathcal{A}\})})\\
&=-H(X_{(2^{\mathcal{S}\cup\mathcal{T}}-2^{\mathcal{T}})}|X,X_{(2^{\mathcal{T}})})+H(X_{(2^{\mathcal{S}}-2^{\mathcal{S}\cap\mathcal{T}})}|X,X_{(2^{\mathcal{S}\cap\mathcal{T}})})+\sum\limits_{\mathcal{A}\in 2^{\mathcal{S}\cup\mathcal{T}}-\mathcal{M}}H(X_{\mathcal{A}}|X_{(2^{\mathcal{A}}-\{\mathcal{A}\})})\\
&\geq-H(X_{(2^{\mathcal{S}\cup\mathcal{T}}-\mathcal{M})}|X,X_{(\mathcal{M})})+\sum\limits_{\mathcal{A}\in 2^{\mathcal{S}\cup\mathcal{T}}-\mathcal{M}}H(X_{\mathcal{A}}|X_{(2^{\mathcal{A}}-\{\mathcal{A}\})})\\
&\geq-\sum\limits_{k=1}^{|\mathcal{S}\cup\mathcal{T}|}\sum\limits_{\mathcal{A}\in 2^{\mathcal{S}\cup\mathcal{T}}-\mathcal{M},|\mathcal{A}|=k}H(X_{\mathcal{A}}|X,\{X_{\mathcal{B}}\}_{\mathcal{B}\in 2^{\mathcal{S}\cup\mathcal{T}},|\mathcal{B}|<k})+\sum\limits_{\mathcal{A}\in 2^{\mathcal{S}\cup\mathcal{T}}-\mathcal{M}}H(X_{\mathcal{A}}|X_{(2^{\mathcal{A}}-\{\mathcal{A}\})})\\
&\geq-\sum\limits_{k=1}^{|\mathcal{S}\cup\mathcal{T}|}\sum\limits_{\mathcal{A}\in 2^{\mathcal{S}\cup\mathcal{T}}-\mathcal{M},|\mathcal{A}|=k}H(X_{\mathcal{A}}|X_{(2^{\mathcal{A}}-\{\mathcal{A}\})})+\sum\limits_{\mathcal{A}\in 2^{\mathcal{S}\cup\mathcal{T}}-\mathcal{M}}H(X_{\mathcal{A}}|X_{(2^{\mathcal{A}}-\{\mathcal{A}\})})\\
&= 0,
\end{align*}
\end{enumerate}
where $\mathcal{M}=2^{\mathcal{S}}\cup 2^{\mathcal{T}}$. The proof is complete.

\section{Proof of Theorem \ref{th:VKG*}}
\label{sec:VKG*}

It is clear that $\mathcal{RD}_{\text{VKG*}}\subseteq\mathcal{RD}_{\text{VKG}}$. Therefore, we just need to show that $\mathcal{RD}_{\text{VKG}}\subseteq\mathcal{RD}_{\text{VKG*}}$.

In view of Lemma \ref{le:contra} and the property of contra-polymatroid \cite{Edmonds70}, for fixed $p_{XX_{(2^{\mathcal{I}_L})}}$ and $\phi_{\mathcal{K}}$, $\emptyset\subset\mathcal{K}\subseteq\mathcal{I}_L$, the region specified by (\ref{eq:rate}) and (\ref{eq:distortion}) has $L!$ vertices:  $(R_1(\pi),\cdots,R_L(\pi),D_{\mathcal{K}}(\pi),\emptyset\subset\mathcal{K}\subseteq\mathcal{I}_L)$ is a vertex for each permutation $\pi$ on $\mathcal{I}_L$, where
\begin{align*}
&R_{\pi(1)}(\pi)=\psi(\{\pi(1)\}),\\
&R_{\pi(k)}(\pi)=\psi(\{\pi(1),\cdots,\pi(k)\})-\psi(\{\pi(1),\cdots,\pi(k-1)\}),\quad k\in\mathcal{I}_L-\{1\},\\
&D_{\mathcal{K}}(\pi)=\mathbb{E}[d(X,\phi_{\mathcal{K}}(X_{(2^{\mathcal{K}})}))],\quad\emptyset\subset\mathcal{K}\subseteq\mathcal{I}_L.
\end{align*}
Since the VKG* region is a convex set, it suffices to show that these $L!$ vertices are contained in the VKG* region.

Without loss of generality, we shall assume that $\pi(k)=k$, $k\in\mathcal{I}_L$. In this case, we have
\begin{align*}
R_L(\pi)=\psi(\mathcal{I}_L)-\psi(\mathcal{I}_{L-1}).
\end{align*}
Now we proceed to write $R_L(\pi)$ as a sum of certain mutual information quantities.
Define
\begin{align}\nonumber
&\mathcal{S}_1(k) = \{ \mathcal{A} : \mathcal{A} \in
\mathcal{L}, |\mathcal{A}| = k, L \in \mathcal{A}\}
\\\nonumber
&\mathcal{S}_2(k) = \{ \mathcal{A} : \mathcal{A} \in
\mathcal{L}, |\mathcal{A}| < k, L \in \mathcal{A}\}.
\end{align}
Note that
\begin{align}\nonumber
R_L(\pi)  & = I(X;X_\emptyset  ) + H(X_{(2^{\mathcal{I}_{L-1}} )} |X) -
H(X_{(2^{\mathcal{I}_L} )} |X) + \sum\limits_{k = 1}^L
{\sum\limits_{\mathcal{A} \in \mathcal{S}_1(k)} {H(
{X_{\mathcal{A}} |X_{(2^{\mathcal{A}} - \{ \mathcal{A}\} )} })} }
\\\nonumber
& = I(X;X_\emptyset ) - H(X_{(2^{\mathcal{I}_L} )} |X,
X_{(2^{\mathcal{I}_{L-1}} )} ) + H(X_{\{L\}}|X_{\emptyset})
+\sum\limits_{k = 2}^L {\sum\limits_{\mathcal{A} \in
\mathcal{S}_1(k)} {H( {X_{\mathcal{A}} |X_{(2^{\mathcal{A}} -
\{ \mathcal{A}\} )} } )} }
\\\nonumber
& = I(X;X_\emptyset  ) + I(X;X_{\{L\}}|X_{(2^{\mathcal{I}_{L-1}})}) +
I(X_{(2^{\mathcal{I}_{L-1}})};X_{\{L\}}|X_{\emptyset}) -
H(X_{(2^{\mathcal{L}})}|X,X_{(2^{\mathcal{I}_{L-1}})}, X_{\{L\}})
\\\nonumber
& \quad + \sum\limits_{k = 2}^L {\sum\limits_{\mathcal{A} \in
\mathcal{S}_1(k)} {H( {X_{\mathcal{A}} |X_{(2^{\mathcal{A}} -
\{ \mathcal{A}\} )} })} }
\\\nonumber
& = I(X;X_\emptyset  ) + I(X;X_{\{L\}}|X_{(2^{\mathcal{I}_{L-1}})}) +
I(X_{(2^{\mathcal{I}_{L-1}})};X_{\{L\}}|X_{\emptyset})
\\\nonumber
& \quad + \sum\limits_{k=2}^L{\left[\sum\limits_{\mathcal{A} \in
\mathcal{S}_1(k)} {H( {X_{\mathcal{A}} |X_{(2^{\mathcal{A}} -
\{ \mathcal{A}\} )} } )} - H(X_{(\mathcal{S}_1(k))}|X,
X_{(2^{\mathcal{I}_{L-1}})}, X_{(\mathcal{S}_2(k))})\right]}.
\end{align}
We arrange the sets in $\mathcal{S}_1(k)$ in some arbitrary order and denote them by $\mathcal{S}_{k,1},\cdots,\mathcal{S}_{k,N(k)}$, respectively, where $N(k) ={{L}\choose{k}} - {{L-1}\choose{k}}$. Then for each $k$,
\begin{align}\nonumber
& \sum\limits_{\mathcal{A} \in \mathcal{S}_1(k)} {H(
{X_{\mathcal{A}} |X_{(2^{\mathcal{A}} - \{ \mathcal{A}\} )} }
)} - H(X_{(\mathcal{S}_1(k))}|X, X_{(2^{\mathcal{I}_{L-1}})},
X_{(\mathcal{S}_2(k))})
\\\nonumber
& = \sum\limits_{i=1}^{N(k)}\left[ {H( {X_{\mathcal{S}_{k,i}} |X_{(2^{\mathcal{S}_{k,i}}
- \{ \mathcal{S}_{k,i}\} )} })} - H(X_{\mathcal{S}_{k,i}}|X,
X_{(2^{\mathcal{I}_{L-1}})}, X_{(\mathcal{S}_2(k))},X_{(\{\mathcal{S}_{k,j}\}_{j=1}^{i-1})})\right]\\\nonumber
&= \sum\limits_{i=1}^{N(k)} {I(X,
X_{(2^{\mathcal{I}_{L-1}})}, X_{(\mathcal{S}_2(k))},X_{(\{\mathcal{S}_{k,j}\}_{j=1}^{i-1})}; {X_{\mathcal{S}_{k,i}} |X_{(2^{\mathcal{S}_{k,i}}
- \{ \mathcal{S}_{k,i}\} )} })}
\\\nonumber
&=\sum\limits_{i=1}^{N(k)}{I(X_{(2^{\mathcal{I}_{L-1}})},X_{(\mathcal{S}_2(k))},X_{(\{\mathcal{S}_{k,j}\}^{i-1}_{j=1})};{X_{\mathcal{S}_{k,i}}|X_{(2^{\mathcal{S}_{k,i}}-\{\mathcal{S}_{k,i}\})}})}+\sum\limits_{i=1}^{N(k)}{I(X;X_{\mathcal{S}_{k,i}}|X_{(2^{\mathcal{I}_{L-1}})}X_{(\mathcal{S}_2(k))}X_{(\{\mathcal{S}_{k,j}\}^{i-1}_{j=1})})}
\\\nonumber
&=\sum\limits_{i=1}^{N(k)}{I(X_{(2^{\mathcal{I}_{L-1}})},X_{(\mathcal{S}_2(k))},X_{(\{\mathcal{S}_{k,j}\}^{i-1}_{j=1})};{X_{\mathcal{S}_{k,i}}|X_{(2^{\mathcal{S}_{k,i}}-\{\mathcal{S}_{k,i}\})}})}+I(X;X_{(\mathcal{S}_1(k))}|X_{(2^{\mathcal{I}_{L-1}})}X_{(\mathcal{S}_2(k))}).
\end{align}
Therefore, we have
\begin{align}\nonumber
R_L(\pi)&=I(X;X_{\emptyset})+I(X;X_{\{L\}}|X_{(2^{\mathcal{I}_{L-1}})})+I(X_{(2^{\mathcal{I}_{L-1}})};X_{\{L\}}|X_{\emptyset})
\\\nonumber
&\quad+\sum\limits_{k=2}^{L}{\left[\sum\limits_{i=1}^{N(k)}{I({X_{\mathcal{S}_{k,i}};X_{(\{\mathcal{S}_{k,j}\}^{i-1}_{j=1})}X_{(2^{\mathcal{I}_{L-1}})}X_{(\mathcal{S}_2(k))}|X_{(2^{\mathcal{S}_{k,i}}-\{\mathcal{S}_{k,i}\})}})}+I(X;X_{(\mathcal{S}_1(k))}|X_{(2^{\mathcal{I}_{L-1}})}X_{(\mathcal{S}_2(k))})\right]}
\\\nonumber
&=I(X;X_{\emptyset})+I(X;X_{(\mathcal{S}_2(L))},X_{\mathcal{I}_L}|X_{(2^{\mathcal{I}_{L-1}})})+I(X_{(2^{\mathcal{I}_{L-1}})};X_{\{L\}}|X_{\emptyset})
\\\label{eq:mutualInfor}
&\quad+\sum\limits_{k=2}^{L-1}{\left[\sum\limits_{i=1}^{N(k)}{I(X_{(2^{\mathcal{I}_{L-1}})},X_{(\mathcal{S}_2(k))},X_{(\{\mathcal{S}_{k,j}\}^{i-1}_{j=1})};{X_{\mathcal{S}_{k,i}}|X_{(2^{\mathcal{S}_{k,i}}-\{\mathcal{S}_{k,i}\})}})}\right]}.
\end{align}

\begin{figure}[t]
   \begin{center}
      \epsfig{file=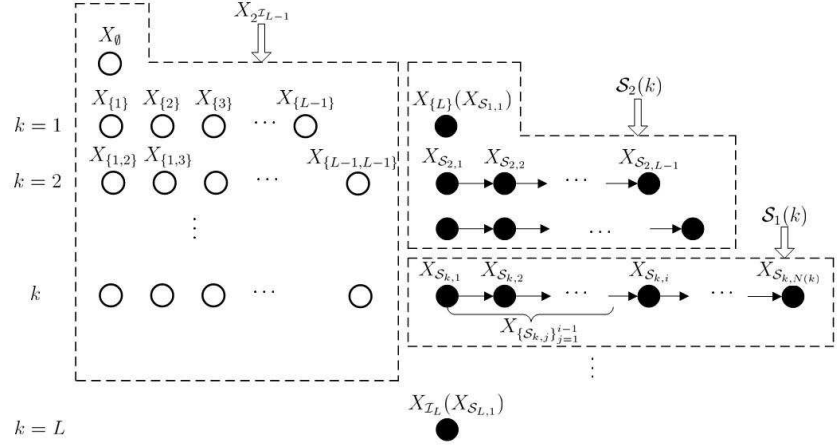, width = 4.5in}
      \caption{The Structure of auxiliary random variables for the VKG region.}
      \label{fig:fig_RV}
    \end{center}
\end{figure}

It follows from Lemma \ref{le:le1} that there exist an auxiliary
random variables $Z$ and a function $f$
such that
\begin{enumerate}
\item $Z$ is independent of
$(X_{(2^{\mathcal{I}_{L-1}})},X_{(\mathcal{S}_2(L))})$;

\item $X_{\mathcal{I}_L} = f(
X_{(2^{\mathcal{I}_{L-1}})},X_{(\mathcal{S}_2(L))},Z)$;

\item $X-(X_{(2^{\mathcal{I}_{L-1}})},X_{(\mathcal{S}_2(L))},X_{\mathcal{I}_L})-Z$ form a Markov chain.
\end{enumerate}
Therefore, we have
\begin{align}\nonumber
&I(X;X_{(\mathcal{S}_2(L))},X_{\mathcal{I}_L}|X_{(2^{\mathcal{I}_{L-1}})})=I(X;X_{(\mathcal{S}_2(L))},Z|X_{(2^{\mathcal{I}_{L-1}})}),
\\\nonumber
&I(X_{(2^{\mathcal{I}_{L-1}})};X_{\{L\}}|X_{\emptyset})=I(X_{(2^{\mathcal{I}_{L-1}})};X_{\{L\}},Z|X_{\emptyset}),
\end{align}
and
\begin{align}
&\nonumber I(X_{(2^{\mathcal{I}_{L-1}})},X_{(\mathcal{S}_2(k))},X_{(\{\mathcal{S}_{k,j}\}^{i-1}_{j=1})};{X_{\mathcal{S}_{k,i}}|X_{(2^{\mathcal{S}_{k,i}}-\{\mathcal{S}_{k,i}\})}})\\
\nonumber&=I(X_{(2^{\mathcal{I}_{L-1}})},X_{(\mathcal{S}_2(k))},X_{(\{\mathcal{S}_{k,j}\}^{i-1}_{j=1})},Z;{X_{\mathcal{S}_{k,i}}|X_{(2^{\mathcal{S}_{k,i}}-\{\mathcal{S}_{k,i}\})}})
\end{align}
for $1 \leq i \leq
N(k)$ and $2 \leq k \leq L-1$.
Now it can be easily verified that $(R_1(\pi),\cdots,R_L(\pi),D_{\mathcal{K}}(\pi),\emptyset\subset\mathcal{K}\subseteq\mathcal{I}_L)$ is preserved if we substitute $X_{\{L\}}$ with $(X_{\{L\}},Z)$, set
$X_{\mathcal{I}_L}$ to be a constant, and modify $\phi_{\mathcal{K}}$, $\emptyset\subset\mathcal{K}\subseteq\mathcal{I}_L$, accordingly. By the definition of the VKG* region, it is clear that $(R_1(\pi),\cdots,R_L(\pi),D_{\mathcal{K}}(\pi),\emptyset\subset\mathcal{K}\subseteq\mathcal{I}_L)\in\mathcal{RD}_{\text{VKG*}}$. The proof is complete.

\section{Proof of Theorem \ref{th:MDindividual}}
\label{sec:MDindividual}

Let $R^*_1=\psi(\{1\})$, and $R^*_k=\psi(\mathcal{I}_k)-\psi(\mathcal{I}_{k-1})$, $k\in\mathcal{I}_L-\{1\}$. By Lemma \ref{le:contra} and the property of contra-polymatroid \cite{Edmonds70}, $(R^*_1,\cdots,R^*_L)$ is a vertex of the rate region $\{(R_1,\cdots,R_L):R_{\mathcal{K}}\geq\psi(\mathcal{K}),\emptyset\subset\mathcal{K}\subseteq\mathcal{I}_L\}$; moreover, we have
\begin{align}
&\min\limits_{(R_1,\cdots,R_L)\in\mathcal{R}_{\text{IH-VKG}}(D_{\mathcal{K}},\mathcal{K}\in\mathcal{H}_L)}\sum\limits_{k=1}^L\alpha_kR_k\nonumber\\
&=\min\limits_{p_{X_{\emptyset}X_{\{1\}}\cdots X_{\{L\}}|X},\phi_{\mathcal{K}},\mathcal{K}\in\mathcal{H}_L}\sum\limits_{k=1}^L\alpha_kR^*_k,\label{eq:min'}
\end{align}
where the minimization in (\ref{eq:min'}) is over $p_{X_{\emptyset}X_{\{1\}}\cdots X_{\{L\}}|X}$, and $\phi_{\mathcal{K}}$, $\mathcal{K}\in\mathcal{H}_L$, subject to the constraints
\begin{align*}
D_{\mathcal{K}}\geq\mathbb{E}[d(X,\phi_{\mathcal{K}}(X_{(2^{\mathcal{K}})}))],\quad \mathcal{K}\in\mathcal{H}_L.
\end{align*}

It follows from Theorem \ref{th:VKG*} that $X_{\mathcal{I}_L}$ can be eliminated. Inspecting (\ref{eq:mutualInfor}) reveals that the same method can be used to eliminate $X_{\mathcal{K}}$, $\mathcal{K}\in\mathcal{S}_2(L)-\{L\}$, successively in the reverse order (i.e., the bottom-to-top and right-to-left order in Fig.\ref{fig:fig_RV}). For $k$ from $L-1$ to $2$, we write $R^*_k$ in a form analogous to (\ref{eq:mutualInfor}) and execute this elimination procedure. In this way all the auxiliary random variables, except $X_{\emptyset},X_{\{1\}},\cdots,X_{\{L\}}$, are eliminated. It can be verified that the resulting expression for $(R^*_1,\cdots,R^*_L)$ is
\begin{align*}
R^*_k=I(X;X_{\emptyset})+I(X,\{X_{\{i\}}\}_{i=1}^{k-1};X_{\{k\}}|X_{\emptyset}),\quad k\in\mathcal{I}_L.
\end{align*}
The proof is complete.

\section{Proof of Corollary \ref{cor:cor3}}\label{app:cor3}

First we shall show that (\ref{eq:min}) is greater than or equal to (\ref{eq:min2}). Let $X'_{\{k\}}=\phi_{\{k\}}(X_{\emptyset},X_{\{k\}})$, $k\in\mathcal{I}_L$, and $X'_{\mathcal{I}_k}=\phi_{\mathcal{I}_k}(X_{\emptyset},X_{\{1\}},\cdots,X_{\{k\}})$, $k\in\mathcal{I}_L-\{1\}$. It can be verified that
\begin{align*}
&\sum\limits_{k=1}^L\alpha_k[I(X;X_{\emptyset})+I(X,\{X_{\{i\}}\}_{i=1}^{k-1};X_{\{k\}}|X_{\emptyset})]\\
&=\sum\limits_{k=1}^L\alpha_k[I(X;X_{\emptyset})+I(\{X_{\{i\}}\}_{i=1}^{k-1};X_{\{k\}}|X_{\emptyset})+I(X;X_{\{k\}}|X_{\emptyset},\{X_{\{i\}}\}_{i=1}^{k-1})]\\
&=\sum\limits_{k=2}^L\alpha_k[I(X;X_{\emptyset})+I(\{X_{\{i\}}\}_{i=1}^{k-1};X_{\{k\}}|X_{\emptyset})]+\sum\limits_{k=1}^L(\alpha_k-\alpha_{k+1})I(X;X_{\emptyset},\{X_{\{i\}}\}_{i=1}^k)\\
&\geq\sum\limits_{k=2}^L\alpha_k[I(X;X_{\emptyset})+I(X'_{(\mathcal{H}_{k-1})};X'_{\{k\}}|X_{\emptyset})]+\sum\limits_{k=1}^L(\alpha_k-\alpha_{k+1})I(X;X_{\emptyset},X'_{(\mathcal{H}_k)})\\
&=\sum\limits_{k=1}^L\alpha_k[I(X;X_{\emptyset})+I(X'_{(\mathcal{H}_{k-1})};X'_{\{k\}}|X_{\emptyset})-I(X;X_{\emptyset},X'_{(\mathcal{H}_{k-1})})+I(X;X_{\emptyset},X'_{(\mathcal{H}_k)})]\\
&=\sum\limits_{k=1}^L\alpha_k[I(X;X_{\emptyset})+I(X'_{(\mathcal{H}_{k-1})};X'_{\{k\}}|X_{\emptyset})+I(X;X'_{\{k\}},X'_{\mathcal{I}_k}|X_{\emptyset},X'_{(\mathcal{H}_{k-1})})],
\end{align*}
where $\alpha_{L+1}\triangleq 0$.

Now we proceed to show that (\ref{eq:min2}) is greater than or equal to (\ref{eq:min}).  It follows from Lemma \ref{le:le1} that there exist a random variable $Z$ and a function $f$ such that
\begin{enumerate}
\item $Z$ is independent of $(X_{\emptyset},X_{(\mathcal{H}_{L-1})},X_{\{L\}})$;

\item $X_{\mathcal{I}_L}=f(X_{\emptyset},X_{(\mathcal{H}_{L-1})},X_{\{L\}},Z)$;

\item $X-(X_{\emptyset},X_{(\mathcal{H}_L)})-Z$ form a Markov chain.
\end{enumerate}
Note that
\begin{align*}
&I(X_{(\mathcal{H}_{L-1})};X_{\{L\}}|X_{\emptyset})=I(X_{(\mathcal{H}_{L-1})};X_{\{L\}},Z|X_{\emptyset}),\\
&I(X;X_{\{L\}},X_{\mathcal{I}_L}|X_{\emptyset},X_{(\mathcal{H}_{L-1})})=I(X;X_{\{L\}},Z|X_{\emptyset},X_{(\mathcal{H}_{L-1})}).
\end{align*}
Therefore, we can substitute $X_{\{L\}}$ with $(X_{\{L\}},Z)$ and eliminate $X_{\mathcal{I}_L}$. It is clear that one can successively eliminate $X_{\mathcal{I}_{L-1}},\cdots,X_{\mathcal{I}_2}$ in a similar manner. The proof is complete.

\section{Proof of Theorem \ref{th:binary}}
\label{sec:binary}

It is obvious that $D^*_{\{2\}}(D_{\{1\}},D_{\{1,2\}})=D_{\{1,2\}}$ if $D_{\{1\}}=\frac{1}{2}$. Therefore, we shall only consider the case $D_{\{1\}}<\frac{1}{2}$.

Since binary symmetric sources are successively refinable, it follows that
\begin{align*}
&R(D_{\{1\}})=1-H_b(D_{\{1\}}),\\
&R(R_1(D_{\{1\}}),D_{\{1\}},D_{\{1,2\}})=1-H_b(D_{\{1,2\}}),
\end{align*}
where $H_b(\cdot)$ is the binary entropy function. If $D_{\{1\}}<\frac{1}{2}$, then $R(D_{\{1\}})$ is achieved if and only if $p_{X_{\{1\}}|X}$ is a binary symmetric channel with crossover probability $D_{\{1\}}$; it is clear that $X$ is not weakly independent with the resulting $X_{\{1\}}$. Therefore, Theorem \ref{th:weakindependence} is applicable here.

Define $X_{\{1,2\}}=g_2(X_{\{1\}},X_{\{2\}})$. Note that we must have $\mathbb{E}[d(X,X_{\{1,2\}})]\leq D_{\{1,2\}}$ and
\begin{align*}
I(X;X_{\{1\}},X_{\{2\}})&=I(X;X_{\{1\}},X_{\{2\}},X_{\{1,2\}})\\
&=I(X;X_{\{1,2\}})\\
&=1-H_b(D_{\{1,2\}}),
\end{align*}
which implies that $X-X_{\{1,2\}}-(X_{\{1\}},X_{\{2\}})$ form a Markov chain and $p_{X_{\{1,2\}}|X}$ is a binary symmetric channel with crossover probability $D_{\{1,2\}}$. Therefore, $p_{XX_{\{1\}}X_{\{1,2\}}}$ is completely specified by the backward test channels shown in Fig. \ref{fig:fig_binary}. Now it is clear that one can obtain $D^*(D_{\{1\}},D_{\{1,2\}})$ by solving the following optimization problem
\begin{align*}
D^*(D_{\{1\}},D_{\{1,2\}})=\min\limits_{p_{X_{\{2\}}|XX_{\{1\}}X_{\{1,2\}}},g_1}\mathbb{E}[d(X,g_1(X_{\{2\}}))]
\end{align*}
subject to the constraints
\begin{enumerate}
\item $X_{\{1\}}$ and $X_{\{2\}}$ are independent;

\item $X_{\{1,2\}}$ is a deterministic function of $X_{\{1\}}$ and $X_{\{1,2\}}$;

\item $X-X_{\{1,2\}}-(X_{\{1\}},X_{\{2\}})$ form a Markov chain.
\end{enumerate}

Assume that $X_{\{2\}}$ takes values in $\{0,1,\cdots,n-1\}$ for some finite $n$. We tabulate $p_{XX_{\{1\}}X_{\{2\}}X_{\{1,2\}}}$, $p_{X_{\{1\}}X_{\{2\}}}$, $p_{XX_{\{2\}}}$, and $p_{X_{\{1\}}X_{\{2\}}X_{\{1,2\}}}$ for ease of reading.
\begin{center}
\begin{tabular}{c|c|c|c|c|c}
 \hline \hline
 \backslashbox{$x,x_{\{1\}},x_{\{1,2\}}$}{$x_{\{2\}}$}&0&1&2&$\cdots$&$n-1$\\
 \hline
 0,0,0&$a_{0,0}$&$a_{0,1}$&$a_{0,2}$&$\cdots$&$a_{0,n-1}$\\
 \hline
 0,0,1&$a_{1,0}$&$a_{1,1}$&$a_{1,2}$&$\cdots$&$a_{1,n-1}$\\
 \hline
 0,1,0&$a_{2,0}$&$a_{2,1}$&$a_{2,2}$&$\cdots$&$a_{2,n-1}$\\
 \hline
 $\vdots$&$\vdots$&$\vdots$&$\vdots$&$\ddots$&$\vdots$\\
 \hline
 1,1,1&$a_{7,0}$&$a_{7,1}$&$a_{7,2}$&$\cdots$&$a_{7,n-1}$\\
 \hline
\end{tabular}
\end{center}

\begin{center} 
\begin{tabular}{c|c|c|c}
 \hline \hline
 \backslashbox{$x_{\{1\}}$}{$x_{\{2\}}$}&0&$\cdots$&$n-1$\\
 \hline
 0&$a_{0,0}+a_{1,0}+a_{4,0}+a_{5,0}$&$\cdots$&$a_{0,n-1}+a_{1,n-1}+a_{4,n-1}+a_{5,n-1}$\\
 \hline
 1&$a_{2,0}+a_{3,0}+a_{6,0}+a_{7,0}$&$\cdots$&$a_{2,n-1}+a_{3,n-1}+a_{6,n-1}+a_{7,n-1}$\\
 \hline
\end{tabular}
\end{center}

\begin{center} 
\begin{tabular}{c|c|c|c}
 \hline \hline
 \backslashbox{$x$}{$x_{\{2\}}$}&0&$\cdots$&$n-1$\\
 \hline
 0&$a_{0,0}+a_{1,0}+a_{2,0}+a_{3,0}$&$\cdots$&$a_{0,n-1}+a_{1,n-1}+a_{2,n-1}+a_{3,n-1}$\\
 \hline
 1&$a_{4,0}+a_{5,0}+a_{6,0}+a_{7,0}$&$\cdots$&$a_{4,n-1}+a_{5,n-1}+a_{6,n-1}+a_{7,n-1}$\\
 \hline
\end{tabular}
\end{center}

\begin{center}
\begin{tabular}{c|c|c|c|c|c|c}
 \hline \hline
 \backslashbox{$x_{\{1,2\}}$}{$x_{\{1\}},x_{\{2\}}$}&0,0&$\cdots$&$0,n-1$&1,0&$\cdots$&$1,n-1$\\
 \hline
 0&$a_{0,0}+a_{4,0}$&$\cdots$&$a_{0,n-1}+a_{4,n-1}$&$a_{2,0}+a_{6,0}$&$\cdots$&$a_{2,n-1}+a_{6,n-1}$\\
 \hline
 1&$a_{1,0}+a_{5,0}$&$\cdots$&$a_{1,n-1}+a_{5,n-1}$&$a_{3,0}+a_{7,0}$&$\cdots$&$a_{3,n-1}+a_{7,n-1}$\\
 \hline
\end{tabular}
\end{center}

\begin{figure}[t]
   \begin{center}
      \epsfig{file=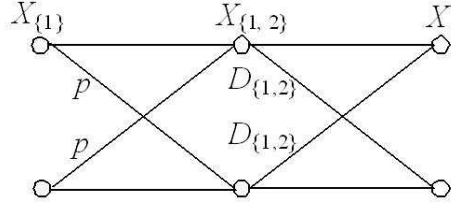, width = 2.5in}
      \caption{The backward channels for successive refinement of a binary symmetric source: $p=\frac{D_{\{1\}}-D_{\{1,2\}}}{1-2D_{\{1,2\}}}$.}
      \label{fig:fig_binary}
    \end{center}
\end{figure}

According to $p_{XX_{\{1\}}X_{\{1,2\}}}$ (cf. Fig. \ref{fig:fig_binary}), it is easy to see that
\begin{align}\nonumber
\sum\limits_{i=0}^{n-1}{a_{0,i}}&=\sum\limits_{i=0}^{n-1}{a_{7,i}}=\frac{1}{2}(1-p)(1-D_{\{1,2\}}),
\\\nonumber
\sum\limits_{i=0}^{n-1}{a_{1,i}}&=\sum\limits_{i=0}^{n-1}{a_{6,i}}=\frac{1}{2}pD_{\{1,2\}},
\\\nonumber
\sum\limits_{i=0}^{n-1}{a_{2,i}}&=\sum\limits_{i=0}^{n-1}{a_{5,i}}=\frac{1}{2}p(1-D_{\{1,2\}}),
\\\label{eq:constraints}
\sum\limits_{i=0}^{n-1}{a_{3,i}}&=\sum\limits_{i=0}^{n-1}{a_{4,i}}=\frac{1}{2}(1-p)D_{\{1,2\}}.
\end{align}

Furthermore, one can verify the following statements.
\begin{enumerate}
\item The fact that $X_{\{1\}}$ and $X_{\{2\}}$ are independent and that $X_{\{1\}}$ is uniformly distributed over $\{0,1\}$  implies
\begin{equation}\label{eq:condition1}
a_{0,i}+a_{1,i}+a_{4,i}+a_{5,i}=a_{2,i}+a_{3,i}+a_{6,i}+a_{7,i},\quad i=0,\cdots,n-1.
\end{equation}

\item The fact that $X_{\{1,2\}}$ is a deterministic function of $(X_{\{1\}},X_{\{2\}})$ implies
\begin{equation}\label{eq:condition2}
(a_{0,i}+a_{4,i})(a_{1,i}+a_{5,i})=(a_{2,i}+a_{6,i})(a_{3,i}+a_{7,i})=0,\quad i=0,\cdots,n-1.
\end{equation}

\item The fact that $X-X_{\{1,2\}}-(X_{\{1\}},X_{\{2\}})$ form a Markov chain implies
\begin{align}\nonumber
a_{0,i}&=\frac{1-D_{\{1,2\}}}{D_{\{1,2\}}}a_{4,i},\quad
a_{5,i}=\frac{1-D_{\{1,2\}}}{D_{\{1,2\}}}a_{1,i}
\\\label{eq:condition3}
a_{2,i}&=\frac{1-D_{\{1,2\}}}{D_{\{1,2\}}}a_{6,i},\quad
a_{7,i}=\frac{1-D_{\{1,2\}}}{D_{\{1,2\}}}a_{3,i},\quad i=0,\cdots,n-1.
\end{align}
\end{enumerate}

According to (\ref{eq:condition2}), there are four
possibilities for each $i$:
\begin{align}\nonumber
a_{0,i}&=a_{2,i}=a_{4,i}=a_{6,i}=0,
\\\nonumber
\text{or} \qquad a_{0,i}&=a_{3,i}=a_{4,i}=a_{7,i}=0,
\\\nonumber
\text{or} \qquad a_{1,i}&=a_{2,i}=a_{5,i}=a_{6,i}=0,
\\\nonumber
\text{or} \qquad a_{1,i}&=a_{3,i}=a_{5,i}=a_{7,i}=0, \quad
i=0,\cdots,n-1.
\end{align}
Moreover, in view of (\ref{eq:condition1}), we can partition $\{0,1,\cdots,n-1\}$ into four
disjoint sets $\mathcal{S}_j$, $j=1,2,3,4$, such that
\begin{align}\nonumber
a_{1,i}+a_{5,i}&=a_{3,i}+a_{7,i},\quad i\in \mathcal{S}_1
\\\nonumber
a_{1,i}+a_{5,i}&=a_{2,i}+a_{6,i},\quad i\in \mathcal{S}_2
\\\nonumber
a_{0,i}+a_{4,i}&=a_{3,i}+a_{7,i},\quad i\in \mathcal{S}_3
\\\label{eq:4set}
a_{0,i}+a_{4,i}&=a_{2,i}+a_{6,i},\quad i\in \mathcal{S}_4.
\end{align}
Combining (\ref{eq:condition3}) and (\ref{eq:4set}) yields
\begin{align}\nonumber
a_{1,i}&=a_{3,i},\;a_{5,i}=a_{7,i},\quad i\in \mathcal{S}_1
\\\nonumber
a_{1,i}&=a_{6,i},\;a_{2,i}=a_{5,i},\quad i\in \mathcal{S}_2
\\\nonumber
a_{0,i}&=a_{7,i},\;a_{3,i}=a_{4,i},\quad i\in \mathcal{S}_3
\\\nonumber
a_{0,i}&=a_{2,i},\;a_{4,i}=a_{6,i},\quad i\in \mathcal{S}_4.
\end{align}
It is easy to see that different values in each $\mathcal{S}_j$, $j=1,2,3,4$, can be combined. That is to say, we can assume that
$X_{\{2\}}$ takes values in
$\{0,1,2,3\}$ with no loss of generality. As a consequence, $p_{XX_{\{1\}}X_{\{2\}}X_{\{1,2\}}}$ and $p_{XX_{\{2\}}}$ can be re-tabulated as follows.
\begin{center}
\begin{tabular}{c|c|c|c|c}
 \hline \hline
 \backslashbox{$x,x_{\{1\}},x_{\{1,2\}}$}{$x_{\{2\}}$}&0&1&2&3\\
 \hline
 0,0,0&0&0&$\beta_3$&$\beta_4$\\
 \hline
 0,0,1&$\alpha_1$&$\alpha_2$&0&0\\
 \hline
 0,1,0&0&$\beta_2$&0&$\beta_4$\\
 \hline
 0,1,1&$\alpha_1$&0&$\alpha_3$&0\\
 \hline
 1,0,0&0&0&$\alpha_3$&$\alpha_4$\\
 \hline
 1,0,1&$\beta_1$&$\beta_2$&0&0\\
 \hline
 1,1,0&0&$\alpha_2$&0&$\alpha_4$\\
 \hline
 1,1,1&$\beta_1$&0&$\beta_3$&0\\
 \hline
\end{tabular}
\end{center}
\begin{center} 
\begin{tabular}{c|c|c|c|c}
 \hline \hline
 \backslashbox{$x$}{$x_{\{2\}}$}&0&1&2&3\\
 \hline
 0&$2\alpha_1$&$\alpha_2+\beta_2$&$\alpha_3+\beta_3$&$2\beta_4$\\
 \hline
 1&$2\beta_1$&$\alpha_2+\beta_2$&$\alpha_3+\beta_3$&$2\alpha_4$\\
 \hline
\end{tabular}
\end{center}
Note that $\alpha_i$ and $\beta_i$ satisfy
\begin{align}\nonumber
&\beta_i=\frac{1-D_{\{1,2\}}}{D_{\{1,2\}}}\alpha_i,\quad i=1,2,3,4,
\\\nonumber
&\alpha_1+\alpha_2=\alpha_4+\alpha_2=\frac{1}{2}pD_{\{1,2\}},
\\\nonumber
&\alpha_1+\alpha_3=\frac{1}{2}(1-p)D_{\{1,2\}},
\end{align}
where the first four equalities follow (\ref{eq:condition3}) while
the others follow (\ref{eq:constraints}). Using $X_{\{2\}}$ to reconstruct $X$, one can achieve
\begin{align}\nonumber
D_{\{2\}}&=2\alpha_1+\alpha_2+\beta_2+\alpha_3+\beta_3+2\alpha_4
\\\nonumber &=\frac{1}{2}-\left(\beta_1-\alpha_1+\beta_4-\alpha_4\right)\\
\nonumber&=\frac{1}{2}-\frac{1-2D_{\{1,2\}}}{D_{\{1,2\}}}\alpha_1-\frac{1-2D_{\{1,2\}}}{D_{\{1,2\}}}\alpha_4.
\end{align}
It can be
easily verified that $D_{\{2\}}$ is minimized when
$\alpha_1=\alpha_4=\frac{1}{2}pD_{\{1,2\}}$. Therefore, we have
\begin{align}\nonumber
D^*_{\{2\}}(D_{\{1\}},D_{\{1,2\}})&=2pD_{\{1,2\}}+\frac{1}{2}(1-2p)\\\nonumber &=\frac{1}{2}+D_{\{1,2\}}-D_{\{1\}}.
\end{align}

\end{document}